\documentclass[journal,comsoc]{IEEEtran}
\usepackage{caption}
\usepackage{subcaption}
\usepackage{lipsum}

\usepackage{color}

\usepackage{graphics}
\usepackage{amsmath,amsfonts, amssymb}
\usepackage{mathtools}
\usepackage{amsthm}
\usepackage{bbm}
\usepackage{algorithm}
\usepackage[noend]{algpseudocode}
\usepackage{dsfont}
\graphicspath{{./}}
\newtheorem*{remark}{Remark}

\ifCLASSINFOpdf
\else
\fi

\begin{document}

\title{Channel Sensing and Communication over a Time-Correlated Channel with an Energy Harvesting Transmitter}

\author{Mehdi~Salehi~Heydar~Abad,
        Ozgur~Ercetin,
        Deniz~G\"und\"uz 
\thanks{ M.S.H.~Abad  and O.~Ercetin are with the Faculty of Engineering and Natural Sciences, Sabanci University, 34956 Istanbul, Turkey.}
\thanks{ D.~G\"und\"uz is with the Department of Electrical and Electronic Engineering, Imperial College London, UK.}
\thanks{ This work was in part supported by EC H2020-MSCA-RISE-2015 programme under grant number 690893, by Tubitak under grant number 114E955 and by British Council Institutional Links Program under grant number 173605884, and by the European Research Council project BEACON under part number 677854.}
}


\maketitle
\newtheorem{theorem}{Theorem}
\newtheorem{lemma}{Lemma}
\newtheorem{corollary}{Corollary}
 
\begin{abstract}

An energy harvesting (EH) transmitter communicating over a time-correlated wireless channel is considered. The transmitter is capable of sensing the current channel state, albeit at the cost of both energy and transmission time. The EH transmitter aims to maximize its long-term throughput by choosing one of the following actions: $i)$ defer its transmission to save energy for future use, $ii)$ transmit reliably at a low rate, $iii)$ transmit at a high rate, and $iv)$ sense the channel to reveal the channel state at a cost of energy and transmission time, and then decide to defer or to transmit. The problem is formulated as a partially observable Markov decision process with a belief on the channel state. The optimal policy is shown to exhibit a threshold behavior on the belief state, with battery-dependent threshold values. The optimal threshold values and performance are characterized numerically via the value iteration algorithm as well as a policy search algorithm that exploits the threshold structure of the optimal policy. Our results demonstrate that, despite the associated time and energy cost, sensing the channel intelligently to track the channel state improves the achievable long-term throughput significantly as compared to the performance of those protocols lacking this ability as well as the one that always senses the channel.

\end{abstract}

\IEEEpeerreviewmaketitle

\section{Introduction}
Due to the tremendous increase in the number of battery-powered wireless communication devices over the past decade, replenishing the batteries of these devices by harvesting energy from natural resources has become an important research area \cite{Paradiso}. Regardless of the type of energy harvesting (EH) device and the energy source employed, a main concern for such communication systems is the stochastic nature of the EH process \cite{EHcont,EHdos,deviller,yang}. To model the uncertainty in the EH process, we consider a  discrete-time system model in which random amount of energy is harvested by the transmitter at each time slot with independent and identically distributed (i.i.d.) values over time\footnote{Typically, the EH process is neither memoryless nor discrete, and the energy is accumulated continuously over time.  However, in order to develop the analytical model underlying this paper, we follow the common assumption in the literature \cite{EHcont,discreteEH}, and assume that the continuous energy arrival is accumulated in an intermediate energy storage device to form energy quantas.}. We assume that the harvested energy is stored in a finite capacity rechargeable battery.

The communication takes place over a  time-varying wireless channel with memory. The channel memory is modeled as a finite state Markov chain \cite{markov}, such that the channel state in the next time slot depends only on the current state. A convenient and often-employed simplification is a two-state Markov chain, known as the Gilbert-Elliot channel \cite{gilbert}. This model assumes that the channel can be either in a \emph{GOOD} or a \emph{BAD} state. We assume that, by spending a certain amount of energy from its battery in a GOOD state, the transmitter can transmit $R_2$ bits of information within a time slot, while in a BAD state, it can only transmit $R_1$ bits, where $R_1<R_2$.


In this work, differently from most of the literature on EH systems, we take into account the energy cost of acquiring channel state information (CSI). At the beginning of each time slot, without the current CSI, EH transmitter takes one of the following actions: $i)$ defer the transmission to save its energy for future use; $ii)$ transmit at a low rate of $R_1$ bits while guaranteeing successful delivery; $iii)$ transmit at a high rate of $R_2$ bits and risk an unsuccessful transmission if the channel is in a BAD state, and $iv)$ sense the channel state, with some time and energy cost, and then decide either to defer or transmit at a rate according to the revealed channel state. Our objective is to maximize the expected discounted sum of bits transmitted over an infinite time horizon.


\subsection{Related Work}
\label{sec:RelatedWork}
Markov decision process (MDP) tools have been extensively utilized in the literature to model communication systems with EH devices. In \cite{MDP1}, the authors propose a simple single-threshold policy for a solar-powered sensor operating over a fading wireless channel. The optimality of a single-threshold policy is proven in \cite{zorzi} when an EH transmitter sends packets with varying importance. The allocation of energy for collecting and transmitting data in an EH communication system is studied in \cite{MDP2} and \cite{MDP3}. The scheduling of EH transmitters with time-correlated energy arrivals  to optimize the long term sum throughput is investigated in \cite{deniz}. Finite time horizon throughput optimization is addressed in \cite{MDP4}, when either the current or future energy and channel states are known by the transmitter. In \cite{MDP5}, power allocation to maximize the throughput is studied when the amount of harvested energy and  channel states are modeled as Markov and static processes, respectively. In \cite{can}, an energy management scheme for sensor nodes with limited energy being replenished at a variable rate is developed to make the probability of complete depletion of the battery arbitrarily small, which at the same time asymptotically maximizes a utility function (e.g., Gaussian channel capacity) that depends on the energy consumption scheme. In \cite{shaviv} a simple online power allocation scheme is proposed for communication over a quasi-static fading channel with an i.i.d. energy arrival process, and it is shown to achieve  the optimal long-term average throughput within a constant gap. Finally, in \cite{mehdi}, a threshold-based transmission scheme over a multiple access channel with no feedback is
investigated when the EH processes are spatially correlated.


Gilbert-Elliott channel model has been previously investigated in the context of scheduling an EH transmitter in \cite{efremidus}, where the transmitter always has perfect CSI, obtained by sensing at every time slot. The transmitter makes a decision to defer or to transmit based on the current CSI and the battery state. Similarly, without considering the channel sensing capability, \cite{Aprem} addresses the problem of optimal power management for an EH sensor over a multi-state wireless channel with memory. Unlike previous work, we take into account the energy cost of channel sensing which can be significant for a low-power EH transmitter. Therefore, in order to minimize the energy consumed for channel sensing, an EH transmitter does not necessarily sense the channel at every time slot, and instead, it  keeps an updated belief of the channel state according to its past observations, and only occasionally senses the current channel state.  

Channel sensing is an essential part of opportunistic and cognitive spectrum access. In \cite{sense1}, the authors investigate the problem of optimal access to a Gilbert-Elliot channel, wherein an energy-unlimited transmitter senses the channel at every time slot. In \cite{gilbertmakale} channel sensing is done only occasionally. The transmitter can decide to transmit at a high or a low rate without sensing the channel, or it can first sense the channel and transmit at a reduced rate due to the time spent for sensing. However, the energy cost of sensing is ignored in \cite{gilbertmakale}. Energy cost of channel sensing has been previously studied in \cite{denizgunduz} for a multiple-input single-output fading channel without memory when both the transmitter and the receiver harvest energy.


\begin{figure}[t]
  \centering
    \includegraphics[scale=.8]{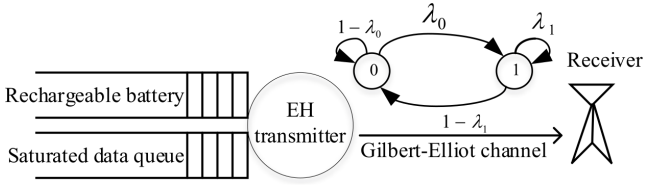}
		  \caption{System model.}
			\label{EHmodel}
\end{figure}

\section{System Model}
\label{sec:SystemModel}
\subsection{Channel and energy harvesting models}
\label{sec:ChannelModelAndEnergyHarvestingAssumptions}

We consider the communication system illustrated in Fig. \ref{EHmodel}, where an EH transmitter communicates over a slotted Gilbert-Elliot channel. Let $G_t$ denote the state of the channel at time slot $t$, which is modeled as a one-dimensional Markov chain with two states: GOOD state denoted by $1$, and BAD state denoted by $0$. Channel transitions occur at the beginning of each time slot. The transition probabilities are given by $\mathds{Pr}\left[G_t=1|G_{t-1}=1\right]=\lambda_1$ and $\mathds{Pr}\left[G_t=1|G_{t-1}=0\right]=\lambda_0$. 
We consider a simple constant-power transmitter which can employ error correcting codes at two different rates, each designed to achieve (almost) reliable transmission at one of the channel states. Accordingly, the transmitter is able to transmit $R_2$ bits per time slot if $G_t=1$, and $R_1<R_2$ bits if $G_t=0$. We normalize the slot duration to one unit; and hence, $R_1$ and $R_2$ refer to both the transmission rate and the number of transmitted bits in a time slot. We assume that the transmitter has an infinitely backlogged data queue, and thus, it always has data to transmit.

We consider an energy \emph{quanta}, representing the smallest energy unit, and assume that the energy arrivals and expenditures, both for transmission and channel sensing, are always integer multiples of this energy unit. At the end of time slot $t$, $E_t$ units of energy arrive according to an i.i.d. random process\footnote{There is an enormous body of the literature (see, for example, \cite{efremidus}, \cite{ind}, and references therein) which assumes i.i.d. EH processes. Nevertheless, results presented in this work can be extended to time-correlated EH processes by incorporating the EH process state into the state of the system. We restrict our attention to i.i.d. EH processes for the clarity of the exposition.}, where $E_t\in\left\{0,1,\ldots,M-1\right\}$ and $\mathds{P}\left[E_t=m\right]=q_m$ for all $t$. The transmitter stores the energy packets in a battery with a capacity of $B_{max}$ units of energy. We denote the state of the battery, i.e., the energy available in the battery at the beginning of time slot $t$, by $B_t$.

\subsection{Transmission protocol}
\label{sec:TransmissionProtocol}
Once a transmission occurs, the receiver replies with an acknowledgment (ACK) if the transmission is successful, or with a negative acknowledgment (NACK) if the transmission fails. Note that, after a transmission at rate $R_2$ an ACK message informs the transmitter that the most recent state of the channel was GOOD, whereas a NACK message informs otherwise. No such information is acquired following a transmission at rate $R_1$, which is successful independent of the channel state.

At the beginning of each time slot, the transmitter takes one of the following actions: $i)$ defer transmission, $ii)$ transmit at rate $R_1$, $iii)$ transmit at rate $R_2$, and $iv)$ sense the channel and transmit or defer, based on the channel state. 

$i)$ \emph{Defer transmission ($D$):} 
The transmitter remains idle, saving its energy to avoid future energy outages. If this action is chosen, there is no message exchange between the transmitter and the receiver. Hence, the transmitter does not obtain the current CSI\footnote{The scenario in which the transmitter is informed about the current CSI even when it does not transmit any data packet is equivalent to the system model investigated in \cite{efremidus}.}.

$ii)$ \emph{Transmit at rate $R_1$ ($L$):} The transmitter transmits at rate  $R_1$ without sensing the channel. If this action is chosen, the transmitter uses a high redundancy coding scheme to guarantee the successful delivery of the message. Since the delivery of the information is guaranteed, the receiver always sends an ACK feedback, and thus, the transmitter does not obtain the current CSI.

$iii)$ \emph{Transmit at rate $R_2$ ($H$):} The transmitter transmits at rate $R_2$ without sensing the channel. If the channel is in a GOOD state, the transmission is successful and the receiver sends an ACK. Otherwise, the transmission fails, and the receiver sends a NACK. This feedback allows the transmitter to obtain the CSI for the completed time slot. We assume that the energy cost of both $L$ and $H$ actions is $\mathcal{E}_T\in \mathbb{Z}^+$ units of energy.

$iv)$ \emph{Channel sensing/Defer at BAD state $OD$:} The transmitter decides to sense the channel at the beginning of the time slot. Channel sensing operation is carried out by sending a control/probing packet, to which the receiver responds with a single bit indicating the channel state. We assume that sensing takes $\tau$ portion of a time slot, where $0<\tau<1$, and the transmitter consumes on average the same power as data transmission over the sensing period. Therefore, the energy cost of channel sensing is $\mathcal{E}_S=\tau\mathcal{E}_T$ units of energy, where $\mathcal{E}_S\in  \mathbb{Z}^+$, and $\mathcal{E}_S<\mathcal{E}_T$. After sensing the channel, if the channel is revealed to be in a GOOD state, in the remaining $1-\tau$ portion of the time slot, the transmitter transmits at rate $R_2$ if it has more than $(1-\tau)\mathcal{E}_T$ energy remaining in the battery. A total of $(1-\tau)R_2$ bits can be transmitted by the end of the time slot. If the channel is revealed to be in a BAD state, then the transmitter defers transmission, saving the rest of the energy (i.e., $(1-\tau)\mathcal{E}_T$).

$v)$ \emph{Channel sensing/Transmit at BAD state $OT$:} The transmitter  again senses the channel initially, and transmits at rate $R_2$ if the channel is in a GOOD state. However, if the channel is in a BAD state, it transmits at rate $R_1$ in the remainder of the time slot.

\begin{remark}
Note that, in both actions involving channel sensing ($OD$ and $OT$) the transmitter transmits at rate $R_2$ if the channel is revealed to be in a GOOD state. This follows from the fact that transmitting at rate $R_2$ when the channel is known to be in a GOOD state has the highest reward for the amount of energy used. A more rigorous proof of this argument is provided in Appendix \ref{reviewerblahblah}.
\end{remark}

Thanks to the channel sensing capability, the transmitter can adapt its behavior to the current channel state. As we show in this paper, this proves to be an important capability to improve the efficiency in EH networks with scarce energy sources.


\section{POMDP Formulation}
\label{sec:MDPFormulation}

At the beginning of each time slot, the transmitter chooses an action from the action set $\mathcal{A}\triangleq\left\{D,L,OD,OT,H\right\}$, based on the state of its battery and its belief about the channel state to maximize a long-term discounted reward to be defined shortly. Although the transmitter is perfectly aware of its battery state, it does not know the current channel state. Hence, the problem can be formulated as a POMDP.

Let the state of the system at time $t$ be denoted by $S_{t}= (B_{t},X_{t})$, where $X_t$ denotes the $\mathit{belief}$ of the transmitter at time slot $t$ about the channel state. The belief $X_{t}$, is the conditional probability that the channel is in a GOOD state at the beginning of the current slot, given the history $\mathcal{H}_{t}$, i.e., $X_{t}=\mathds{Pr}\left[G_t=1|\mathcal{H}_{t}\right]$, where $\mathcal{H}_{t}$ represents all past actions and observations of the transmitter up to, but not including, slot $t$. The belief of the transmitter constitutes a sufficient statistic to characterize its optimal actions \cite{lovejoy}. Note that with this definition of a state, the POMDP problem is converted into an MDP with an uncountable state space $\left\{0,1,2,\ldots,B_{max}\right\}\times\left[0,\ 1\right]$.


A transmission policy $\pi$ describes a set of rules that dictate which action to take at each slot  depending on the history. Let $V^{\pi}(b,p)$ be the expected infinite-horizon discounted reward with initial state $S_0=(b,\ \mathds{Pr}\left[G_0=1|\mathcal{H}_{0}\right]=p)$ under policy $\pi$ with discount factor $\beta\in[0,\ 1)$. The use of the expected discounted reward allows us to obtain a tractable solution, and one can gain insights into the optimal policy for the average reward when $\beta$ is close to 1. $\beta$ can be interpreted as the probability that the transmitter is allowed to use the channel, or  the probability of the transmitter to remain active at each time slot as in \cite{learningdeniz}. For an initial belief $p$, the expected discounted reward has the following expression
\begin{align}
V^{\pi}(b,\ p) = \mathds{E}\left[\sum^{\infty}_{t=0}\beta^{t}R(S_{t},A_{t})|S_{0}=(b,\ p)\right],\label{Vdef}
\end{align}
where $t$ is the time index, $A_{t}\in\mathcal{A}$ is the action chosen at time $t$, and $R(S_{t},A_{t})$ is the expected reward acquired when action $A_t$ is taken at state $S_{t}$. The expectation in (\ref{Vdef}) is over the state sequence distribution induced by the given transmission policy $\pi$. The expected reward when action $A_t$ is chosen at state $S_t$ is given as follows:

\begin{align}
R(S_{t},A_{t}) = \left\{
\begin{array}{ll}
X_{t}R_2, &  A_t = H,\ B_t\geq \mathcal{E}_T,\\
R_1, &   A_t = L,\ B_t\geq \mathcal{E}_T,\\
(1-\tau)X_tR_2, &   A_t = OD,\ B_t\geq \mathcal{E}_T,\\
(1-\tau)[(1-X_t)R_1\\
\hspace{1cm}+X_tR_2 ], & A_t = OT,\ B_t\geq \mathcal{E}_T,\\
0, & \text{otherwise}.
\end{array} \right.\label{eq:Ri}
\end{align}

Since $\mathcal{E}_T$ energy units is required for transmission (with or without channel sensing), if the battery state is below $\mathcal{E}_T$, the reward becomes zero. Hence, in (\ref{eq:Ri}) we only consider actions taken when the battery state is at least $\mathcal{E}_T$. If the action of transmitting at rate $R_2$ without sensing is chosen, $R_2$ bits are transmitted successfully if the channel is in a GOOD state, and $0$ bits otherwise. Since the belief, $X_t$, represents the probability of the channel being in a GOOD state, the expected reward is given by $X_tR_2$. It is guaranteed that transmitting at low rate is always successful, so the expected reward for this action is $R_1$. If the action of channel sensing is chosen, $\mathcal{E}_S=\tau\mathcal{E}_T$  energy units is spent sensing the channel with the remaining $(1-\tau)\mathcal{E}_T$  energy units either being used for transmission, or saved in the battery. If the channel is in a GOOD state, $(1-\tau)R_2$ bits are transmitted successfully.  If the channel is in a BAD state, the transmitter either remains silent and receives no rewards, or utilizes  $(1-\tau)\mathcal{E}_T$ energy units and transmits $(1-\tau)R_1$ bits in the rest of the time slot. Thus, the expected reward of action $OD$ is $(1-\tau)X_tR_2$, while the expected reward of $OT$ is $(1-\tau)[(1-X_t)R_1+X_tR_2 ]$. Finally, if the action of deferring ($D$) is taken, the transmitter neither senses the channel nor transmits data, so the reward is zero.

Define the value function $V(b,\ p)$ as
\begin{align}
V(b,\ p) &= \max_{\pi}V^{\pi}(b,\ p),\ \forall b\in\left\{0,1,\ldots,B_{max}\right\},\ \forall p\in\left[0,\ 1\right]\label{Vmax}.
\end{align}
The optimal infinite-horizon expected reward can be achieved by a stationary policy, i.e., there exists a stationary policy $\pi^*$, mapping the state space $\left\{0,1,\ldots,B_{max}\right\}\times\left[0,\ 1\right]$ into the action space $\mathcal{A}$, such that $V(b,\ p) = V^{\pi^*}(b,\ p)$ \cite{puterman}. The value function $V(b,\ p)$ satisfies the Bellman equation
\begin{align}
V(b,\ p) = \max_{A\in\left\{D,L,OD,OT,H\right\}}\left\{V_{A}(b,\ p)\right\},
\end{align}
where $V_{A}(b,\ p)$ is the action-value function, defined as the expected infinite-horizon discounted reward acquired by taking action $A$ in state $(b,\ p)$, and is given by
\begin{align}
V_{A}(b,\ p)=&R((b,\ p),A)\nonumber\\
&+\beta\mathds{E}_{(\acute{b},\ \acute{p})}\left[V(\acute{b},\ \acute{p})|S_{0}=(b,\ p),A_0=A\right], \label{actionvalue}
\end{align}
where $(\acute{b},\ \acute{p})$ denotes the next state when action $A$ is taken at state $S_0=(b,\ p)$. The expectation in (\ref{actionvalue}) is over the distribution of next states. Below, we evaluate the action-value function $V_{A}(b,\ p)$, and how the system state evolves for each action.


\emph{Defer transmission ($D$):}
Since there is no transmission, there is no feedback; and thus, the transmitter does not learn the  the channel state. Therefore, the belief is updated as the probability of finding the channel in a GOOD state given the current belief state. If $X_t=p$ at time slot $t$, after taking action $D$, belief is updated as
\begin{align}
J(p) = \lambda_{0}(1-p)+\lambda_{1}p. \label{eq:belief}
\end{align}
After taking action $D$, the value function evolves as:
\begin{align}
V_D(b,\ p) &= \beta\sum^{M-1}_{m=0}q_m V(\min(b+m,B_{max}),J(p)).
\end{align}
Note that the term $\min(b+m,B_{max})$ is used to ensure that the battery state does not exceed the battery capacity, $B_{max}$.

\emph{Transmit at rate $R_1$ ($L$):}
This action can be taken only if\footnote{Note that in the generic MDP formulation, we have the same set of actions in every state. We can re-define the reward function by assigning $-\infty$ reward to those actions that are not possible to be taken in specific states to account for this. For the ease of exposition, we chose to present the formulation in this manner.} $b\geq \mathcal{E}_T$. The transmission will be successful independent of the channel state. Hence, the ACK feedback from the receiver does not inform the transmitter about the channel state. Similarly to action $D$, the  belief state is updated using (\ref{eq:belief}), and the value function is given by:

\begin{align}
V_L(b,\ p) &= R_1 + \beta\sum^{M-1}_{m=0}q_m V(\min(b+m-\mathcal{E}_T,B_{max}),J(p)).
\end{align}

\emph{Transmit at rate $R_2$ ($H$):}
This action can only be chosen if $b\geq \mathcal{E}_T$. If the channel is in  GOOD state, $R_2$ bits are successfully delivered to the receiver, the receiver sends back an ACK, and the belief for the next time slot is updated as $\lambda_1$. Otherwise, the transmission fails, the receiver sends a NACK, and the belief is updated as $\lambda_0$. Hence, the value function evolves as:

\begin{align}
V_H(b,\ p) &= p\left[R_2 + \beta\sum^{M-1}_{m=0}q_m V(\min(b+m-\mathcal{E}_T,B_{max}),\lambda_1) \right]\nonumber\\
&+(1-p)\left[ \beta\sum^{M-1}_{m=0}q_m V(\min(b+m-\mathcal{E}_T,B_{max}),\lambda_0) \right].
\end{align}

\emph{Channel sensing/ Defer in BAD state ($OD$):}
If $b\geq \mathcal{E}_T$ and the transmitter decides to sense the channel, it consumes $\mathcal{E}_S=\tau\mathcal{E}_T$ units of energy to sense the current channel state. If the channel is found to be in a GOOD state, $(1-\tau)\mathcal{E}_T$ units of energy is used to transmit $(1-\tau)R_2$ bits, and the belief state is updated as $\lambda_1$. Note that the transmitter always transmits if the channel is in a GOOD state, because this is the best state possible and saving energy for future cannot improve the reward. We refer the interested readers to Appendix \ref{reviewerblahblah} for a rigorous proof of this claim. In action $OD$, transmission is deferred if the channel is in a BAD state, and the transmitter saves $(1-\tau)\mathcal{E}_T$ units of energy for possible future transmissions. The belief is updated as $\lambda_0$ for the next time slot. The action-value function for action $OD$ is given by:
\begin{align}
&V_{OD}(b,\ p)\nonumber\\
&= p\left[(1-\tau)R_2 + \beta\sum^{M-1}_{m=0}q_m V(\min(b+m-\mathcal{E}_T,B_{max}),\lambda_1) \right]\nonumber\\
&+(1-p)\left[ \beta\sum^{M-1}_{m=0}q_m V(\min(b+m-\tau\mathcal{E}_T,B_{max}),\lambda_0) \right].
\end{align}

Meanwhile, if $\tau\mathcal{E}_T\leq b<\mathcal{E}_T$, transmission is not possible since it requires at least $\mathcal{E}_T$ units of energy. However, it is still possible to sense the channel, since it only requires $\tau\mathcal{E}_T$ units of energy. This case may arise when the transmitter believes that learning the channel state may help its decision in the future. Thus, for $\tau\mathcal{E}_T\leq b<\mathcal{E}_T$, the action-value function for action $OD$ is given by:
\begin{align}
V_{OD}(b,\ p)&= p \beta\sum^{M-1}_{m=0}q_m V(\min(b+m-\tau\mathcal{E}_T,B_{max}),\lambda_1)\nonumber\\
&+(1-p) \beta\sum^{M-1}_{m=0}q_m V(\min(b+m-\tau\mathcal{E}_T,B_{max}),\lambda_0). \label{ODOT}
\end{align}

\emph{Channel sensing/Transmit at BAD state ($OT$):} The transmitter senses the channel, and transmits no matter what the channel state is. It transmits $(1-\tau)R_2$ bits if it is in a GOOD state, and $(1-\tau)R_1$ bits in a BAD state. The belief is updated as $\lambda_1$ ($\lambda_0$) if the channel is in a GOOD (BAD) state. The action-value function is given by:

\begin{align}
&V_{OT}(b,\ p)\nonumber\\
&= p\left[(1-\tau)R_2 + \beta\sum^{M-1}_{m=0}q_m V(\min(b+m-\mathcal{E}_T,B_{max}),\lambda_1) \right]\nonumber\\
&+(1-p)\bigg[(1-\tau)R_1 +\beta\sum^{M-1}_{m=0}q_m V(\min(b+m-\mathcal{E}_T,B_{max}),\lambda_0) \bigg].
\end{align}

Based on the action-value functions presented above, the evolution of the battery state is as follows:
\begin{align}
B_{t+1} = \left\{
\begin{array}{ll}
\min(B_t+E_t,B_{max}), &A_t = D,\\
\min(B_t+E_t-\mathcal{E}_T,B_{max}),\hspace{-.2cm} &A_t \in \left\{L,H,OT\right\},B_t\geq \mathcal{E}_T,\\
\min(B_t+E_t-\tau\mathcal{E}_T\\
\hspace{0.8cm}-(1-\tau)\mathcal{E}_T G_t,B_{max}), &A_t = OD,B_t\geq \mathcal{E}_T\\
\min(B_t+E_t-\tau\mathcal{E}_T,B_{max}), &A_t = OD, \tau\mathcal{E}_T\leq b<\mathcal{E}_T.
\end{array} \right.\label{eq:Bt2}
\end{align}

\section{The Structure of The Optimal Policy}
\label{sec:StructureOfTheOptimalPolicy}

\subsection{General Case}
In this section, we show that the optimal policy has a threshold-type structure on the belief state. The belief state set, i.e., the interval $[0,\ 1]$, can be divided into mutually exclusive subsets where each subset is assigned to a distinct action. We begin to establish our main results by proving the convexity of the value function $V(b,p)$, with respect to $p$. 
\begin{lemma}
\label{thm:convex_general}
For any given $b\geq 0$, V(b,\ p) is convex in $p$.
\end{lemma} 
\begin{proof}
The proof is given in Appendix \ref{sec:convex_general}.
\end{proof}
Next, we show that the value function is a non-decreasing function of the battery state, $b$. This lemma provides the intuition why deferring or sensing actions are advantageous in some states. The incentive of taking these actions is that the value function transitions into higher values without consuming any energy, or consuming only $\tau\mathcal{E}_T$ energy units.

\begin{lemma}
\label{thm:nondecreasing-battery}
	 Given an arbitrary belief $0\leq p\leq 1$, $V(b_1,p)\geq V(b_0,p)$ if $b_1>b_0$.
\end{lemma}
\begin{proof}
The proof is given in Appendix \ref{sec:proofofthm2}.
\end{proof}

The next lemma states that the value function is non-decreasing with respect to the belief state as well. 

\begin{lemma}
\label{thm:nondecreasing-belief}
For a given battery state $b\in\left\{0,1,\ldots,B_{max}\right\}$, if $p_{1}>p_{0}$ then $V(b,p_{1})\geq V(b,p_{0})$.
\end{lemma}
\begin{proof}
The proof is given in Appendix \ref{sec:proofofthm3}.
\end{proof}

Lemma \ref{thm:convex_general} is necessary in proving the structure of the optimal policy. For each $b\in\left\{0,1,\ldots,B_{max}\right\}$ and $A\in\mathcal{A}$, we define:
\begin{align}
\Phi^{b}_{A}\triangleq\left\{p\in\left[0,1\right]:V(b,p)=V_{A}(b,p)\right\}.\label{sets}
\end{align}
For any $b\geq 0$, $\Phi^{b}_{A}$ characterizes the set of belief states for which it is optimal to choose action $A$. In Theorem \ref{thm:structure_general}, we show that the optimal policy has a threshold-type structure.
\begin{theorem}
\label{thm:structure_general}
The optimal policy is a threshold-type policy on the belief state $p$, and the thresholds are functions of the battery state, $b$.
\end{theorem} 
\begin{proof}
This theorem states that the optimal policy has a threshold structure. Initially, we aim to prove that $\Phi^{b}_{A}$ for $A\in\left\{OD,OT,H\right\}$ is convex. It is easy to see that for $b=0$, $V(b,p)=V_{D}(b,p)$, and hence, $\Phi^{0}_{D}=[0,\ 1]$, and $\Phi^{0}_{L}=\Phi^{0}_{OD}=\Phi^{0}_{OT}=\Phi^{0}_{H}=\varnothing$. First, we consider battery states $\tau\mathcal{E}_T\leq b<\mathcal{E}_T$. We will prove that for any $\tau\mathcal{E}_T\leq b<\mathcal{E}_T$, $\Phi^{b}_{OD}$ is convex. 
Let $p_1,\ p_2\in \Phi^{b}_{OD}$, and $a\in(0,\ 1)$. We have
\begin{align}
V(b,ap_1+(1-a)p_2)&\leq aV(b,p_1)+(1-a)V(b,p_2),\label{ineq1}\\
&=aV_{OD}(b,p_1)+(1-a)V_{OD}(b,p_2),\label{eq1}\\
&=V_{OD}(b,ap_1+(1-a)p_2),\label{eq2}\\
&\leq V(b,ap_1+(1-a)p_2),\label{ineq2}
\end{align}
where \eqref{ineq1} follows from Lemma \ref{thm:convex_general}; \eqref{eq1} is due to the fact that $p_1,\ p_2\in \Phi^{b}_{OD}$; \eqref{eq2} follows from the linearity of $V_{OD}$ in $p$; and \eqref{ineq2} holds due to the definition of $V(b,\ p)$. Consequently, $V(b,ap_1+(1-a)p_2) = V_{OD}(b,ap_1+(1-a)p_2)$, and it follows that $ap_1+(1-a)p_2\in \Phi^{b}_{OD}$, which, in turn, proves the convexity of $\Phi^{b}_{OD}$. Note also that $p=0$ and $p=1$ both belong to $\Phi^{b}_{D}$ for all $0\leq b<\mathcal{E}_T$. Since no transmission is possible for $0\leq b<\mathcal{E}_T$, we have $\Phi^{b}_{L}=\Phi^{b}_{H} = \varnothing$.  Hence, for $0\leq b<\mathcal{E}_T$, either $\Phi^{b}_{OD}=\varnothing$, or there exists $0<\rho_1(b)\leq\rho_2(b)<1$ such that $\Phi^{b}_{OD}=[\rho_1(b),\rho_2(b)]$. Consequently, we have $\Phi^{b}_{D}=[0,\rho_1(b))\cup(\rho_2(b),1]$, if $0\leq b<\mathcal{E}_T$.

Next, consider $\mathcal{E}_T\leq b \leq B_{max}$. We will prove that $\Phi^{b}_{H}$, $\Phi^{b}_{OD}$, and $\Phi^{b}_{OT}$ are convex subsets of the belief state set. Let $p_1,\ p_2\in \Phi^{b}_{H}$  and $a\in(0,\ 1)$. Similar to (\ref{ineq1})-(\ref{ineq2}) we can argue
\begin{align}
V(b,ap_1+(1-a)p_2)&\leq aV(b,p_1)+(1-a)V(b,p_2),\nonumber\\
&=aV_{H}(b,p_1)+(1-a)V_{H}(b,p_2),\nonumber\\
&=V_{H}(b,ap_1+(1-a)p_2),\nonumber\\
&\leq V(b,ap_1+(1-a)p_2).
\end{align}
Consequently, $V(b,ap_1+(1-a)p_2) = V_{H}(b,ap_1+(1-a)p_2)$; and hence, $ap_1+(1-a)p_2\in \Phi^{b}_{H}$, which proves the convexity of $\Phi^{b}_{H}$. Since it is always optimal to transmit at rate $R_2$ if the channel is in a GOOD state (see \cite{efremidus}, and Appendix \ref{reviewerblahblah}) $1\in\Phi^{b}_{H}$, and since the convex subsets of the real line are intervals, there exists $\rho_{N}(b)\in(0,1]$ such that $\Phi^b_{H} = [\rho_{N}(b),1]$. Note that $N$ is the number of thresholds, which depends on the system parameters. Using the same technique we can prove that $\Phi^{b}_{OD}$ and $\Phi^{b}_{OT}$ are both convex, and hence, there exists $0<\rho_{i_1}(b)\leq\rho_{i_2}(b)\leq\rho_{j_1}(b)\leq\rho_{j_2}(b)\leq \rho_{N}(b)\leq 1$, such that $\Phi^b_{OD}=[\rho_{i_1}(b),\rho_{i_2}(b)]$ and $\Phi^b_{OT}=[\rho_{j_1}(b),\rho_{j_2}(b)]$; or $\Phi^b_{OT}=[\rho_{i_1}(b),\rho_{i_2}(b)]$ and $\Phi^b_{OD}=[\rho_{j_1}(b),\rho_{j_2}(b)]$. However, since $V_{A}(b,ap_1+(1-a)p_2)\neq aV_{A}(b,p_1)+(1-a)V_{A}(b,p_2)$ for  $A\in\left\{D,L\right\}$, in general, $\Phi^{b}_{D}$ and $\Phi^{b}_{L}$ are not necessarily convex sets.
\end{proof}
Although the optimal policy is of threshold-type, as shown in Theorem \ref{thm:structure_general}, the subsets of the belief space associated with actions $D$ and $L$, i.e., $\Phi^{b}_{D}$ and $\Phi^{b}_{L}$, are not necessarily convex. Each of these sets can be composed of infinitely many intervals; therefore, despite the threshold-type structure, characterizing the optimal policy may require identifying infinitely many threshold values. Finding the exact $N$ and corresponding threshold values is elusive and out of the scope of this paper.

\subsection{Special Case: $R_1 = 0$}
\label{SP-case}
In order to further simplify the problem we assume that it is not possible to transmit any bits when the channel is in a BAD state, i.e., $R_1=0$ and $R_2=R$. Hence, action $L$ is no longer available, and the action for sensing the channel consists of  only $OD$ which is denoted by $O$ in the rest of this section.

With this modified model, the expected reward function can be simplified as follows:
\begin{align}
R(S_{t},A_{t}) = \left\{
\begin{array}{ll}
X_{t}R, & \text{if }  A_t = H\ \text{and}\ B_t\geq \mathcal{E}_T,\\
(1-\tau)X_t R, & \text{if }  A_t = O\ \text{and}\ B_t\geq \mathcal{E}_T,\\
0, & \text{otherwise}.
\end{array} \right.\label{eq:R}
\end{align}

Since at least $\mathcal{E}_T$ energy units is required for transmission, if  $b<\mathcal{E}_T$, the reward in (\ref{eq:R}) becomes zero. If action $H$ is taken, $R$ bits are transmitted successfully if the channel is in a GOOD state, and 0 bits otherwise. If action $O$ is taken, $\tau\mathcal{E}_T$ energy units is spent sensing the channel with the remainder of the energy being used for transmission if the channel is in a GOOD state. In this case, $(1-\tau)R$ bits are transmitted successfully. If the channel is in a BAD state, the transmitter remains silent in the rest of the time slot. Finally, if action $D$ is taken the reward is zero.

Next, we prove that the optimal policy has a threshold-type structure on the belief state with a finite number of thresholds. Note that, in the modified model, the value function is still convex and Lemmas \ref{thm:convex_general}, \ref{thm:nondecreasing-battery} and \ref{thm:nondecreasing-belief} still hold. Theorem \ref{thm:threshold} below states that the optimal solution of the problem defined in (\ref{Vmax}) is a threshold-type policy with either two or three thresholds on the belief state. Threshold values depend on the state of the battery and system parameters.


\begin{theorem}
\label{thm:threshold}
Let $p\in[0,\ 1]$ and $b\geq 0$. There are thresholds $0\leq\rho_1(b)\leq\rho_2(b)\leq\rho_3(b)\leq 1$, all of which are functions of the battery state $b$, such that for $\tau\mathcal{E}_T\leq b< \mathcal{E}_T$

\begin{align}
\pi^{*}(b,p)
&= \left\{
\begin{array}{rl}
D, & \text{if }\ 0\leq p<\rho_1(b)\ \text{or}\ \rho_2(b)< p\leq 1,\\
O, & \text{if }\ \rho_1(b)\leq p\leq\rho_2(b). \label{threshb}
\end{array} \right.
\end{align}
and for $b\geq \mathcal{E}_T$,
\begin{align}
\pi^{*}(b,\ p)
&= \left\{
\begin{array}{rl}
D, & \text{if }\ 0\leq p<\rho_1(b)\ \text{or}\ \rho_2(b)< p<\rho_3(b)\\
O, & \text{if }\ \rho_1(b)\leq p\leq\rho_2(b),\\
H, & \text{if }\ \rho_3(b)\leq p\leq 1, \label{thresha}
\end{array} \right.
\end{align}

\end{theorem}
\begin{proof}
The proof follows similarly to the proof of Theorem \ref{thm:structure_general}. Consider the sets $\Phi^{b}_{A}$ defined in (\ref{sets}) for $A\in\left\{D,O,H\right\}$. 
Note that for $b=0$, $V(b,p)=V_{D}(b,p)$, and hence, $\Phi^{0}_{D}=[0,\ 1]$, and $\Phi^{0}_{O}=\Phi^{0}_{H}=\varnothing$. First, consider battery states $\tau\mathcal{E}_T\leq b<\mathcal{E}_T$. We prove that for any $\tau\mathcal{E}_T\leq b<\mathcal{E}_T$, $\Phi^{b}_{O}$ is convex, which implies the structure of the optimal policy in (\ref{threshb}).
Let $p_1,\ p_2\in \Phi^{b}_{O}$, and $a\in(0,\ 1)$. We have
\begin{align}
V(b,ap_1+(1-a)p_2)&\leq aV(b,p_1)+(1-a)V(b,p_2),\label{ineq1s}\\
&=aV_{O}(b,p_1)+(1-a)V_{O}(b,p_2),\label{eq1s}\\
&=V_{O}(b,ap_1+(1-a)p_2),\label{eq2s}\\
&\leq V(b,ap_1+(1-a)p_2),\label{ineq2s}
\end{align}
where \eqref{ineq1s} follows from Lemma \ref{thm:convex_general}; \eqref{eq1s} is due to the fact that $p_1,\ p_2\in \Phi^{b}_{O}$; \eqref{eq2s} follows from the linearity of $V_{O}$ in $p$; and \eqref{ineq2s} from the definition of $V(b,\ p)$. Hence, $V(b,ap_1+(1-a)p_2) = V_{O}(b,ap_1+(1-a)p_2)$, and it follows that $ap_1+(1-a)p_2\in \Phi^{b}_{O}$, which, in turn, proves the convexity of $\Phi^{b}_{O}$. Note also that $p=0$ and $p=1$ both belong to $\Phi^{b}_{D}$ for all $0\leq b<\mathcal{E}_T$. Hence, for $0\leq b<\mathcal{E}_T$, either $\Phi^{b}_{O}=\varnothing$, or there exists $0<\rho_1(b)\leq\rho_2(b)<1$ such that $\Phi^{b}_{O}=[\rho_1(b),\rho_2(b)]$. Consequently, we have $\Phi^{b}_{D}=[0,\rho_1(b))\cup(\rho_2(b),1]$.

Next, consider $\mathcal{E}_T\leq b \leq B_{max}$. We prove that $\Phi^{b}_{H}$ and $\Phi^{b}_{O}$ are both convex, which implies the structure of the optimal policy in (\ref{thresha}). Let $p_1,\ p_2\in \Phi^{b}_{H}$  and $a\in(0,\ 1)$. Similarly to (\ref{ineq1})-(\ref{ineq2}) we can argue
\begin{align}
V(b,ap_1+(1-a)p_2)&\leq aV(b,p_1)+(1-a)V(b,p_2),\nonumber\\
&=aV_{H}(b,p_1)+(1-a)V_{H}(b,p_2),\nonumber\\
&=V_{H}(b,ap_1+(1-a)p_2),\nonumber\\
&\leq V(b,ap_1+(1-a)p_2).
\end{align}
Thus, $V(b,ap_1+(1-a)p_2) = V_{H}(b,ap_1+(1-a)p_2)$; and hence, $ap_1+(1-a)p_2\in \Phi^{b}_{H}$, which proves the convexity of $\Phi^{b}_{H}$. Since it is always optimal to transmit at rate $R_2$ if the channel is in a GOOD state, $1\in\Phi^{b}_{H}$, and since the convex subsets of the real line are intervals, there exists $\rho_{3}(b)\in(0,1]$ such that $\Phi^b_{H} = [\rho_{3}(b),1]$. Using the same technique we can prove that $\Phi^{b}_{O}$ is convex; and hence, there exists $0<\rho_1(b)\leq\rho_2(b)<1$ such that $\Phi^b_{O}=[\rho_1(b),\rho_2(b)]$. The remaining segments belong to action $D$, and we have $\Phi_{D}=[0,\rho_1(b))\cup(\rho_2(b),\rho_3(b))$.
\end{proof}
Theorem \ref{thm:threshold} proves that at any battery state $b\geq \mathcal{E}_T$, at most three threshold values are sufficient to characterize the optimal policy; whereas two thresholds suffice for $0\leq b<\mathcal{E}_T$. However the optimal policy can even be simpler for some battery states and some instances of the problem as it is possible to have $\rho_2(b) = \rho_3(b)$, or even $\rho_1(b) = \rho_2(b) = \rho_3(b)$. Since, $\Phi^{b}_{D}$ is not a convex set in general (see Theorem \ref{thm:convex_general}), the structure of the optimal policy may result in four different regions even though there are only three possible actions. This may seem counter intuitive since deferring the transmission should not be advantageous when the belief is relatively high. Nevertheless, in Section \ref{num}, we demonstrate that in some cases it is indeed optimal to have a three-threshold policy.



\section{Numerical Results}
\label{num}
In this section, we use numerical techniques to characterize the optimal policy, and evaluate its performance. We utilize the value iteration algorithm to calculate the optimal value function. We numerically identify the thresholds for the optimal policy for different scenarios. We also evaluate the performance of the optimal policy, and compare it with some alternative policies in terms of throughput. 

\subsection{Evaluating the optimal policy}
In the following, we consider the modified system model introduced in Section \ref{SP-case} in which no data can be transmitted in a BAD channel state, i.e., $R_1=0$. Moreover, without loss of generality, we set $M=11$, $\mathcal{E}_T=10$, and $q_{10}=q=1-q_0$ and $q_m=0$ for $m=1,\ldots,9$. We assume that $B_{max}=50$, $\tau=0.2$, $\beta=0.98$, $\lambda_1=0.9$, $\lambda_0=0.6$, $R = 3$ and $q=0.1$. The optimal policy is evaluated using the value iteration algorithm. In Fig. \ref{5region}, each state $(b,\ p)$ is illustrated with a different color corresponding to the optimal policy at that state. In Fig. \ref{5region}, the areas highlighted with blue correspond to those states at which deferring the transmission is optimal, green areas correspond to the states at which sensing the channel is optimal, and finally yellow areas correspond to the states at which transmitting at high rate is optimal. As seen in Fig. \ref{5region}, depending on the battery state the optimal policy may have one, two, or three thresholds on the belief state. For example, when the battery state is $b=20$, there is a single threshold; the transmitter defers transmission up to a belief state of $p=0.8$, and starts transmitting without sensing beyond this value. For no value of the belief state it opts for sensing the channel. On the other hand, when the battery state is $38$, the policy has two thresholds, and three thresholds when the battery state is $28$. Considering the low probability of energy arrivals ($q=0.1$) and the relative high cost of sensing ($\tau=0.2$), the transmitter senses the channel even when its battery state is below the transmission threshold, i.e., $b<10$. 

\begin{figure}[ht]
  \centering
    \includegraphics[scale=.6]{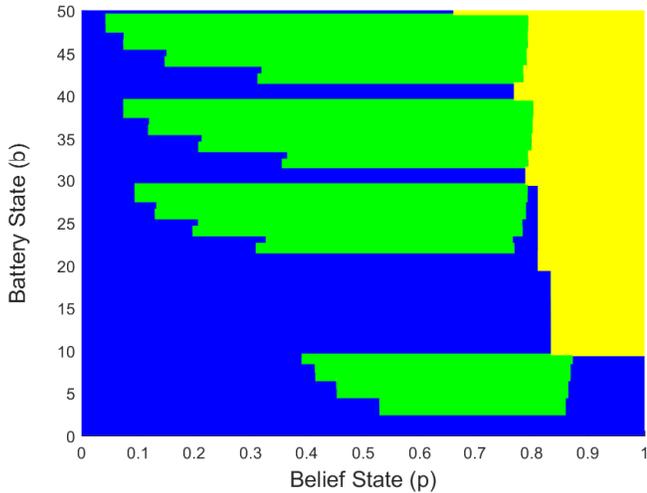}
		  \caption{Optimal thresholds for taking the actions D (blue), O (green), H (yellow) for $B_{max}=50$, $\mathcal{E}_T = 10$, $\tau=0.2$, $\beta=0.98$, $\lambda_1=0.9$, $\lambda_0=0.6$, $R = 3$ and $q=0.1$.}
			\label{5region}
\end{figure}

Another interesting observation from Fig. \ref{5region} is the periodicity of the optimal policy with respect to the battery. This is particularly visible for action $D$ taken when the battery state is an integer multiple of $\mathcal{E}_T$, which is then followed by action $O$ for increasing beliefs when  the battery state is more than 20. The value function corresponding to the parameters used to obtain Fig. \ref{5region} is depicted in Fig. \ref{valuefunction}. Note the staircase behavior of the value function. There is a jump in the value function when the battery state is an integer multiple of $\mathcal{E}_T$, while it approximately remains the same when the battery state is confined between two consecutive integer multiples of $\mathcal{E}_T$, i.e., ($n\mathcal{E}_T\leq b < (n+1)\mathcal{E}_T$), where $n$ is an integer. Hence, when the battery state of the transmitter is an integer multiple of $\mathcal{E}_T$, any action other than deferring will, with high probability, transition into a state with a relatively lower value. Thus, the transmitter chooses action $D$ unless its belief is relatively high. However, when the battery state is between two consecutive integer multiples of $\mathcal{E}_T$, it is safe to sense the channel, since, in the worst case, the channel is  in a BAD state and the transmitter loses only $\tau\mathcal{E}_T<\mathcal{E}_T$ units, but it makes a transition into a state which approximately has the same value. Thus, at those values of the battery, the transmitter senses the channel for moderate belief states.

\begin{figure}[ht]
  \centering
    \includegraphics[scale=0.6]{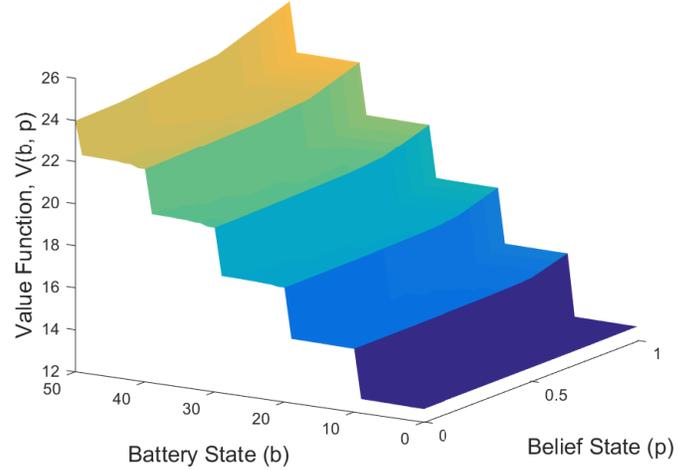}
		  \caption{Value function associated with $B_{max}=50$, $\mathcal{E}_T=10$, $\tau=0.2$, $\beta=0.98$, $\lambda_1=0.9$, $\lambda_0=0.6$, $R = 3$ and $q=0.1$.}
			\label{valuefunction}
\end{figure}

To investigate the effect of the EH rate, $q$, on the optimal transmission policy, we consider the system parameters $B_{max}=50$, $\tau=0.1$, $\beta=0.9$, $\lambda_1=0.8$, $\lambda_0=0.4$, and $R = 3$. We illustrate the optimal transmission policy for $q=0.8$ and $q=0.2$ in Fig. \ref{7region} and Fig. \ref{9region}, respectively. It can be observed by comparing those two figures that the yellow regions are much larger and blue areas are much more limited in Fig. \ref{7region}. This is because when the energy arrivals are more frequent, the EH node tends to consume its energy more generously. We also observe that the transmitter always defers its transmission for $b<10$ when energy is limited (in Fig. \ref{9region}), whereas it may opt for sensing the channel when energy is more abundant.

\begin{figure}
        \centering
        \begin{subfigure}[t]{0.49\textwidth}
                \includegraphics[width=\textwidth]{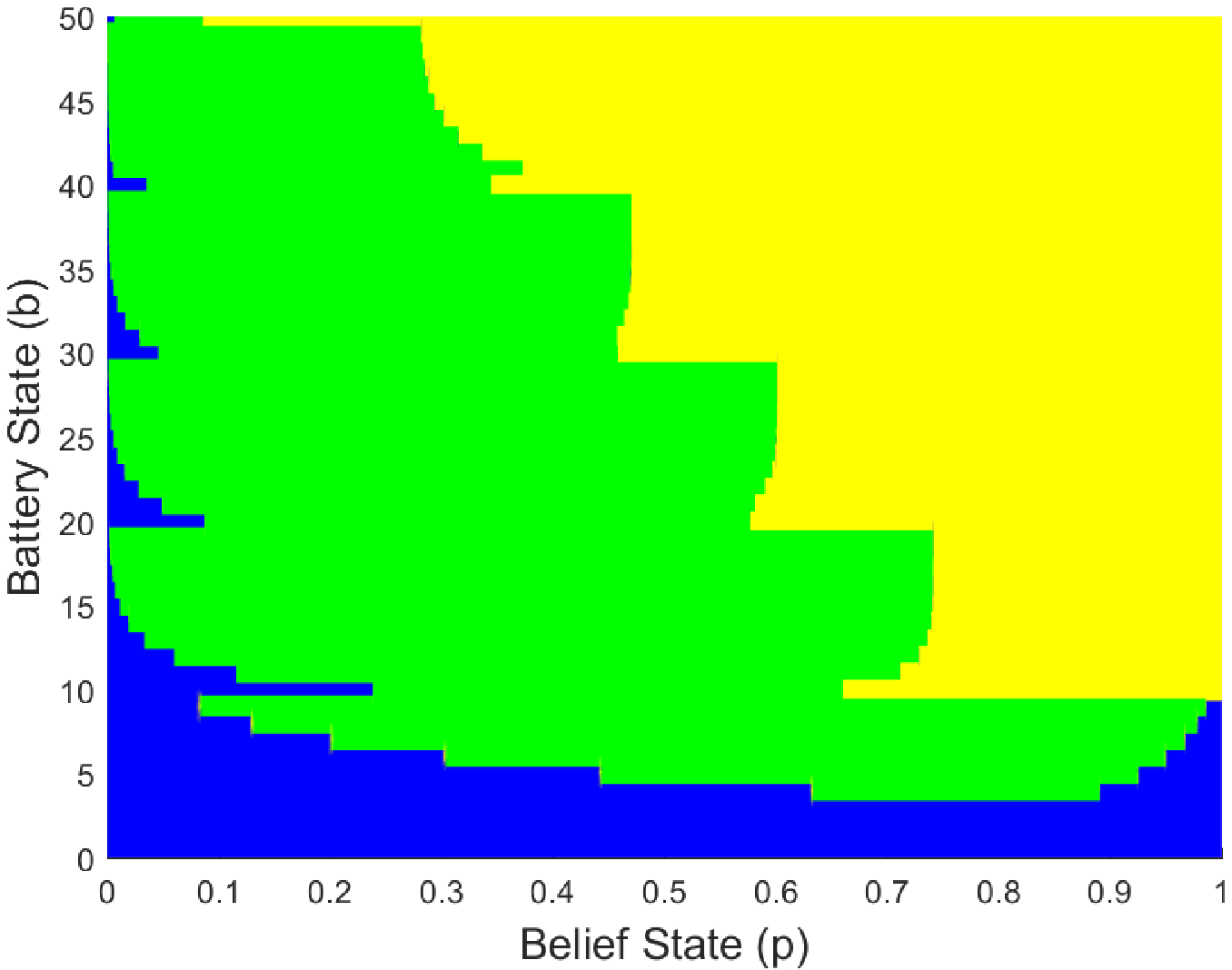}
                \caption{$q =0.8$.}
                \label{7region}
        \end{subfigure}%
        \hfill
        \begin{subfigure}[t]{0.49\textwidth}
                \includegraphics[width=\textwidth]{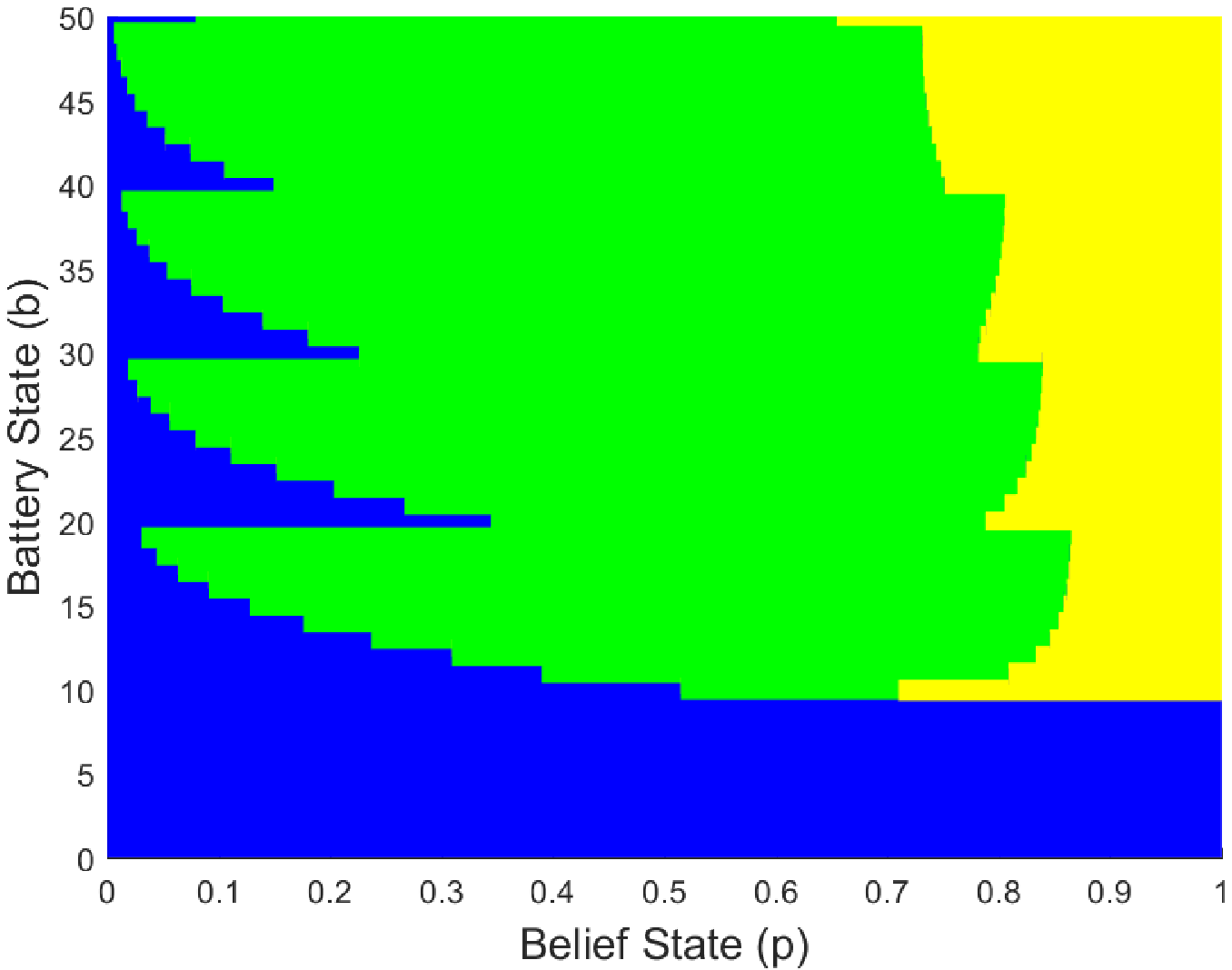}
                \caption{$q =0.2$. }
                \label{9region}
        \end{subfigure}
				\caption{Optimal thresholds for taking the actions D (blue), O (green), H (yellow) for $B_{max}=50$, $\mathcal{E}_T = 10$, $\tau=0.1$, $\beta=0.9$, $\lambda_1=0.8$, $\lambda_0=0.4$, and $R = 3$.}
\label{qregion}
\end{figure}



Next, we investigate the effect of the sensing cost, $\tau$, on the optimal policy. We set the system parameters as $B_{max}=50$, $\beta=0.9$, $\lambda_1=0.8$, $\lambda_0=0.4$, $R = 3$ and $q=0.8$. The regions for optimal actions are shown in Fig. \ref{10region} and Fig. \ref{12region} for sensing cost values $\tau=0.2$ and $\tau = 0.3$, respectively. By comparing Fig. \ref{10region} and Fig. \ref{12region}, it is evident that a higher cost of sensing reduces the incentive for sensing the channel. We observe in Fig. \ref{12region} that the green areas have shrunk as compared to Fig. \ref{10region}, i.e, the transmitter is more likely to take a risk and transmit without sensing, or defer its transmission, when sensing consumes a significant portion of the available energy.

\begin{figure}
        \centering
        \begin{subfigure}[t]{0.49\textwidth}
                \includegraphics[width=\textwidth]{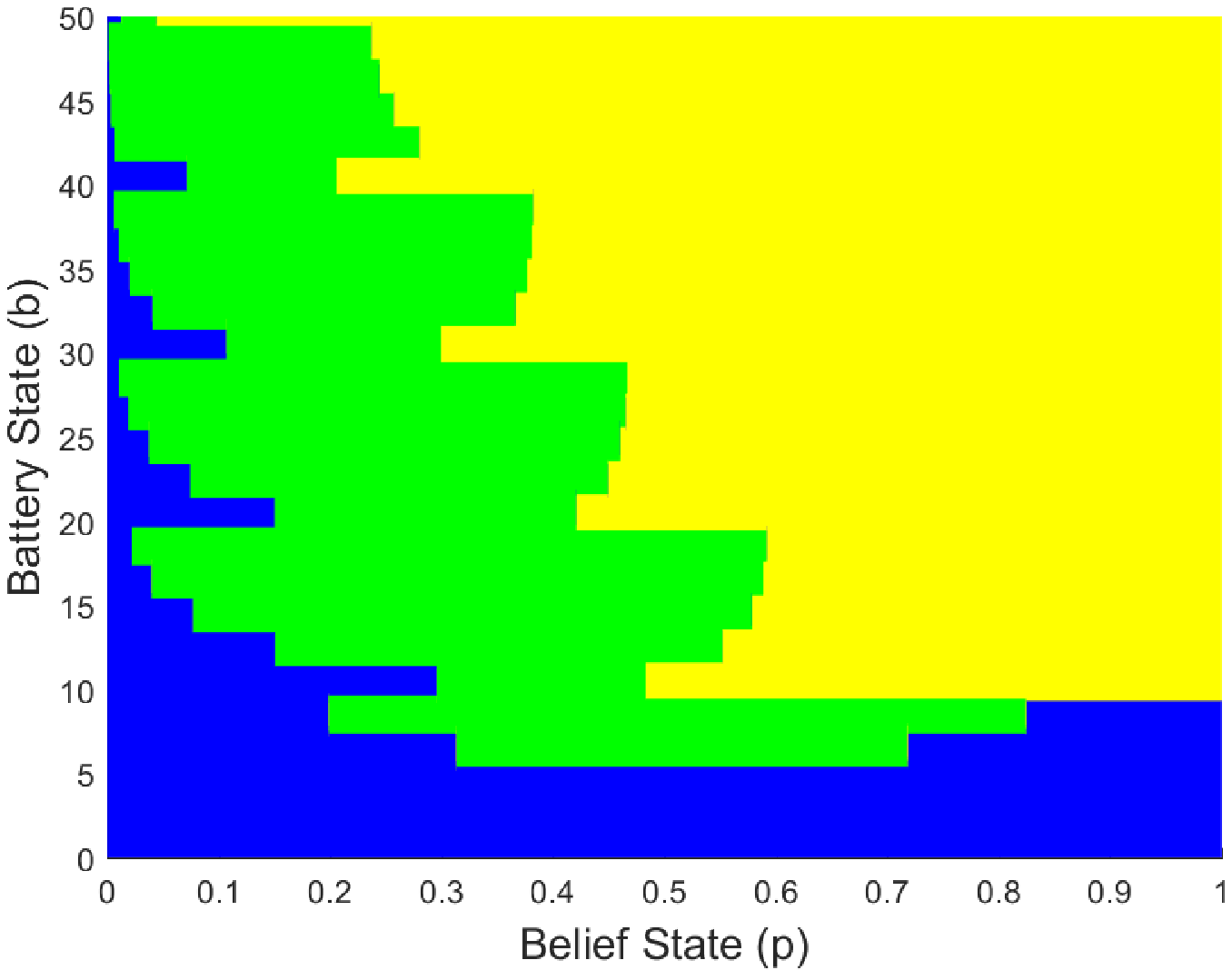}
                \caption{$\tau =0.2$.}
                \label{10region}
        \end{subfigure}
       \hfill
        \begin{subfigure}[t]{0.49\textwidth}
                \includegraphics[width=\textwidth]{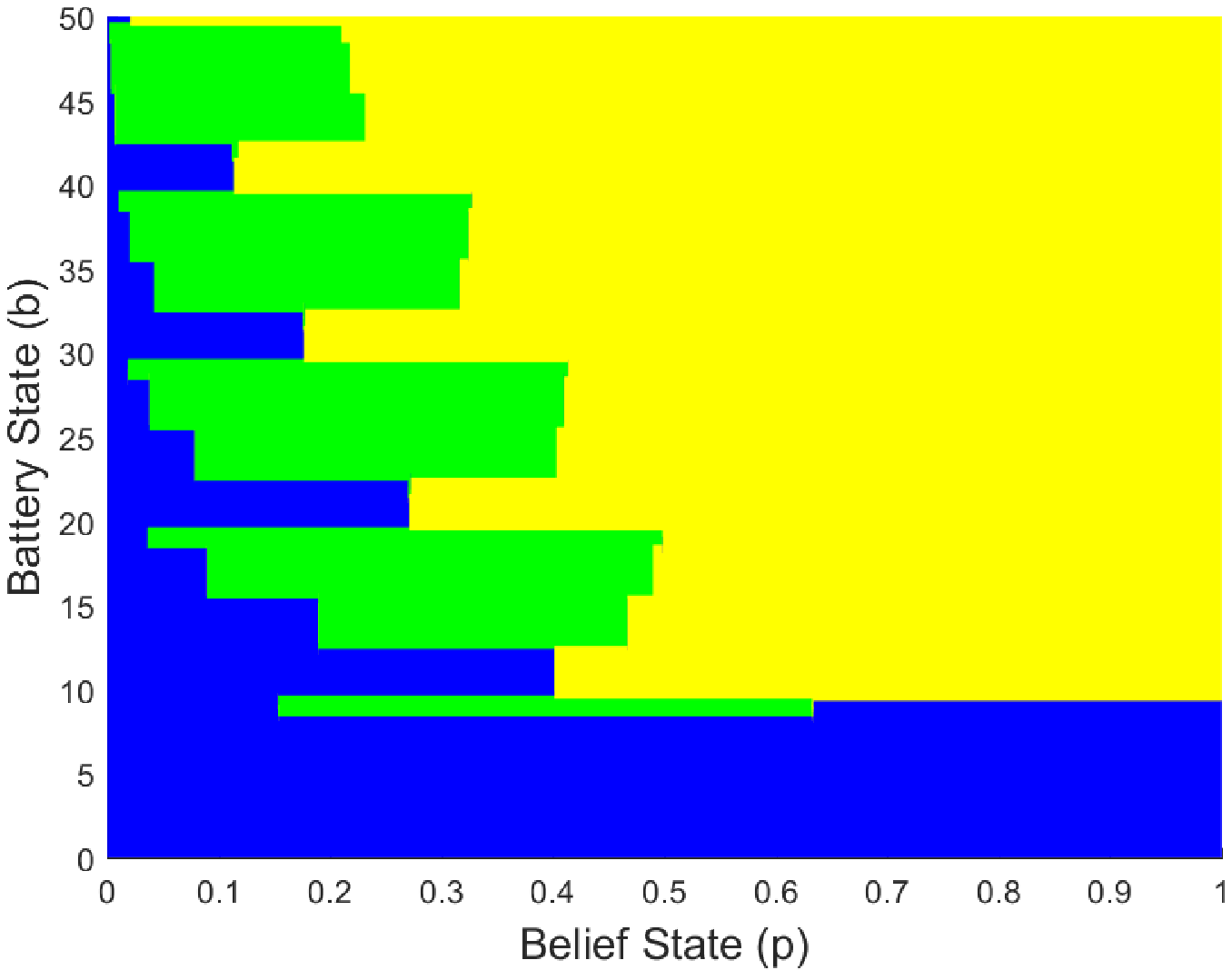}
                \caption{$\tau =0.3$. }
                \label{12region}
        \end{subfigure}
				\caption{Optimal thresholds for taking the actions D (blue), O (green), H (yellow) for $B_{max}=50$, $\beta=0.9$, $\mathcal{E}_T = 10$, $\lambda_1=0.8$, $\lambda_0=0.4$, $R = 3$ and $q=0.8$.}
\label{tauregion}
\end{figure}



\subsection{Throughput performance}
In this section, we compare the performance of the optimal policy with three alternative policies, i.e., a greedy policy, a single-threshold policy and an opportunistic policy. For the optimal policy, as an alternative to the value iteration algorithm, we also employ \emph{policy search} approach, which exploits the threshold structure of the optimal policy that we have proven. For the value iteration algorithm, the average discounted reward is evaluated with a discount value close to 1 ($\beta=0.999$) to approximate the optimal average throughput. Note that, the value iteration algorithm does not exploit the structure of the optimal policy and uses action-value functions to maximize the discounted reward. The policy search method \cite{psearch}, on the other hand, uses the structure of the optimal policy, and the thresholds are directly optimized to maximize the average throughput (and not the discounted throughput). In the \emph{greedy policy}, the EH node transmits whenever it has energy in its battery. In the \emph{single-threshold policy}, there are only two actions: defer (D) or transmit (H). The belief of the transmitter on the current channel state depends only on the ACK/NACK feedback from the receiver, and channel sensing is not exploited at all. We optimize the threshold corresponding to each battery state for the single-threshold policy using the value iteration algorithm. Meanwhile, the \emph{opportunistic policy} senses the channel at the beginning of every time slot, and transmits $(1-\tau)R$ bits if the channel is in a GOOD state, and defers otherwise. By choosing the parameters $B_{max}=50$, $\mathcal{E}_T = 10$, $\beta=0.999$, $\lambda_1=0.8$, $\lambda_0=0.2$, $R = 2$, the  throughput achieved by these four policies are plotted in Fig. \ref{fig:TH} with respect to the EH rate $q$. Fig. \ref{fig:TH1} and Fig. \ref{fig:TH2} correspond to the sensing costs of $\tau=0.1$ and $\tau=0.2$, respectively.

\begin{figure}
        \centering
        \begin{subfigure}[t]{0.5\textwidth}
                \includegraphics[width=\textwidth]{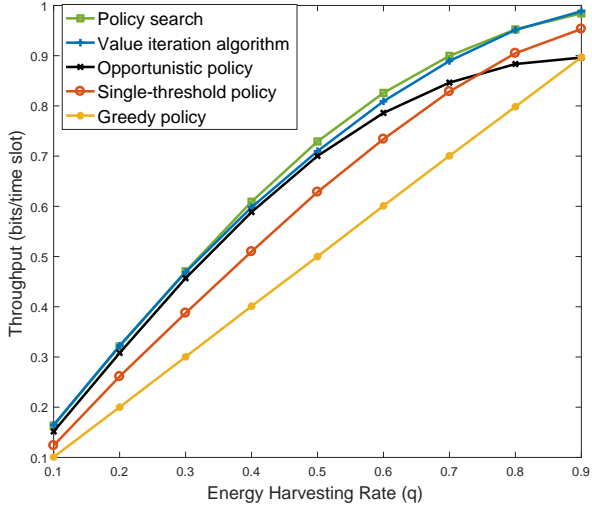}
                \caption{$\tau =0.1$.}
                \label{fig:TH1}
        \end{subfigure}%
        \hfill
        \begin{subfigure}[t]{0.5\textwidth}
                \includegraphics[width=\textwidth]{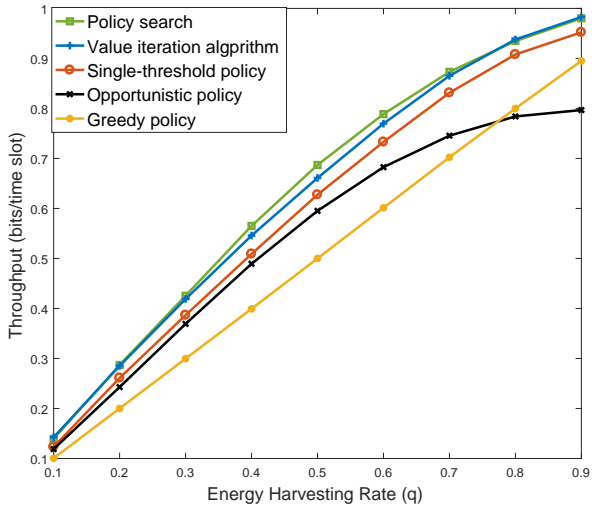}
                \caption{$\tau =0.2$. }
                \label{fig:TH2}
        \end{subfigure}
				\caption{Throughputs by the optimal, greedy, single-threshold and opportunistic policies as a function of the EH rate, $q$.}
\label{fig:TH}
\end{figure}

As expected, the greedy policy performs much worse than the optimal policy as it does not exploit the transmitter's knowledge about the state of the channel. We can see that, by simply exploiting  the ACK/NACK feedback from the receiver in order to defer transmission, the single-threshold policy already achieves a significantly higher throughput than the greedy policy at all values of the EH rate. Note that single-threshold and greedy policies do not have the sensing capability, and accordingly, the sensing cost, $\tau$, has no effect on their performance. However, $\tau$ affects the optimal and opportunistic policies which have sensing capabilities. In particular, $\tau$ affects the opportunistic policy drastically, since this policy senses the channel at the beginning of each time slot. When the sensing cost is relatively low, it can be seen from Fig. \ref{fig:TH1} that the opportunistic policy achieves a near optimal throughput except when the EH rate, $q$, is high. For high values of $q$, the EH transmitter suffers less from energy deprivations and instead of sensing at each time slot, using the whole time slot for transmission becomes more beneficial. Hence, we observe that always sensing the channel performs poorly $q$ is high. When $\tau$ is relatively high, it can be seen from Fig. \ref{fig:TH2} that the opportunistic policy performs worse than the single-threshold policy for all values of $q$, and even worse than the greedy policy for high values of $q$. On the other hand, the optimal policy, by intelligently utilizing the sensing capability, yields a superior performance for all the parameter values.

\begin{remark}
We remark that the policy search achieves a better performance than the value iteration algorithm. This is because the latter maximizes the discounted reward rather than the average reward. To obtain the optimal average reward using value iteration algorithm, we need to set $\beta \to 1$. However, the value iteration algorithm is computationally demanding, and letting $\beta \to 1$ deteriorates its convergence rate to the point of infeasibility. On the other hand, policy search optimizes the thresholds directly to maximize the average throughput, and it is much faster compared to the value iteration algorithm. We owe this superior performance to the structure of the optimal policy that we have shown.
\end{remark}

\subsection{Optimal policy evaluation with two different transmission rates}
When the transmitter has the ability to transmit at two different rates, we proved that the optimal policy is a
threshold-type policy; however, due to non-convexity of sets $\Phi^{b}_{D}$ and $\Phi^{b}_{L}$ it is not possible to characterize the optimal policy as we have done for a transmitter with a single rate in (\ref{thresha}) and (\ref{threshb}). Instead, we numerically evaluate the optimal policy as follows.

Let $B_{max}=5$, $\mathcal{E}_T=200$, $\mathcal{E}_S=7$, $\beta=0.7$, $\lambda_1=0.98$, $\lambda_0=0.81$, $R_1 = 2.91$, $R_2=3$ and $q_{201}=q=1-q_0$ and $q_m=0$ for $m=1,\ldots,200$. Note that these parameters are chosen in a way to show the non-convexity of the sets $\Phi^{b}_{D}$ and $\Phi^{b}_{L}$ and may not be relevant for a practical scenario.
The optimal policy, obtained through the value iteration algorithm, is represented in Fig. \ref{2rate}. In the figure, the areas highlighted with blue correspond to the states at which deferring ($D$) is optimal, red correspond to states at which transmitting at the low rate ($L$) is optimal, green correspond to states at which sensing and deferring is optimal ($OD$), black correspond to  states at which sensing and transmitting opportunistically ($OT$) is optimal, and yellow correspond to the states for which transmitting without sensing ($H$) is optimal.

\begin{figure}[ht]
  \centering
    \includegraphics[scale=.6]{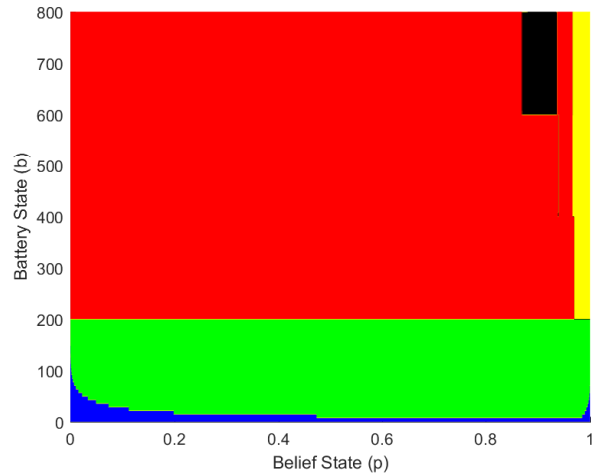}
		  \caption{Optimal thresholds for taking the actions D (blue), L (red), OD (green), OT (black), H (yellow) for $B_{max}=800$, $\mathcal{E}_T=200$, $\tau=0.035$, $\beta=0.7$, $\lambda_1=0.98$, $\lambda_0=0.81$, $R_1 = 2.91$, $R_2=3$ and $q=0.1$.}
			\label{2rate}
\end{figure}
As expected the optimal policy is again a  battery-dependent threshold-type policy with respect to the belief state. The sets $\Phi^{b}_{D}$ and $\Phi^{b}_{L}$ (blue and red areas, respectively) are not convex. In theory, an optimal policy may have infinite threshold values if the sets $\Phi^{b}_{D}$ and $\Phi^{b}_{L}$ are intertwined into infinitely many alternating intervals. We observe in Fig. \ref{2rate} that, for the parameters considered here, this is not the case and the optimal policy consists of at most three-threshold policies.

\section{Conclusions}
\label{concl}
In this work, we considered an EH transmitter equipped with a finite-capacity battery, operating over a time-varying finite-capacity channel with memory, modeled as a two-state Gilbert-Elliot channel. The transmitter receives ACK/NACK feedback after each transmission, which can be used to track the channel state.  Additionally, the transmitter has the capability to sense the channel, which allows the transmitter to obtain the current channel state at a certain energy and time cost. Therefore, at the beginning of each time slot, the transmitter has the following possible actions to maximize the total expected discounted number of bits transmitted over an infinite time horizon: $i)$ deferring transmission, $ii)$ transmitting at a low rate
of $R_1$
bits with guaranteed successful delivery, $iii)$ transmitting at a high rate of $R_2$
bits, and $iv)$ sensing the channel to reveal the channel state by consuming a portion of its energy and
transmission time, and then deciding either to defer or to transmit at a suitable rate based on the channel state. We formulated the problem as a POMDP, which is then converted into an MDP with continuous state space by introducing a belief parameter for the channel state. We have shown that the optimal transmission policy has a threshold structure with respect to the belief state, where the optimal threshold values depend on the battery state.

We then considered the simplified problem by assuming that it is not possible to transmit any information when the channel is in a BAD state, for which we were able to prove that the optimal policy has at most three thresholds. We calculated the optimal threshold values numerically using the value iteration and policy search algorithms. We compared the throughput achieved by the optimal policy to those achieved by a greedy policy and a single-threshold policy, which do not exploit the channel sensing capability, as well as an opportunistic policy, which senses the channel at every time slot. We have shown through simulations that the intelligent channel sensing capability improves the performance significantly, thanks to the increased adaptability to channel conditions. 


\begin{appendices}

\section{Proof of Lemma \ref{thm:convex_general}}
\label{sec:convex_general}
\allowdisplaybreaks
\begin{proof}
Define $V(b,p,n)$ as the optimal value function for the finite-horizon problem spanning only $n$ time slots. We will first prove the convexity of $V(b,\ p,\ n)$ in $p$ by induction. Optimal value function can be written as follows,
\begin{align}
V(b,\ p,\ n) = \max\left\{\right.&V_{D}(b,\ p,n),V_{L}(b,\ p,n),V_{OD}(b,\ p,n),\nonumber\\
& V_{OT}(b,\ p,n),V_{H}(b,\ p,n)\left.\right\},
\end{align}
where
\begin{align}
V_{D}(b,\ p,\ n)=& \beta\sum^{M-1}_{m=0}q_m V(\min(b+m,B_{max}),J(p),\ n-1)
,\label{VDn}
\end{align}
\begin{align}
&V_{L}(b,\ p,\ n)=R_1  \nonumber\\
&\hspace{0.5cm}+ \beta\sum^{M-1}_{m=0}q_m V(\min(b+m-\mathcal{E}_T,B_{max}),J(p),\ n-1)\label{VLn},
\end{align}
\begin{align}
&V_{OD}(b,\ p,\ n) = p\bigg[(1-\tau)R_2\nonumber\\
&+ \beta\sum^{M-1}_{m=0}q_m V(\min(b+m-\mathcal{E}_T,B_{max}),\lambda_1,\ n-1) \bigg]\nonumber\\
&+(1-p)\left[ \beta\sum^{M-1}_{m=0}q_m V(\min(b+m-\tau\mathcal{E}_T,B_{max}),\lambda_0,\ n-1) \right],\nonumber\\
&\hspace{6cm}\text{for}\ \  b\geq \mathcal{E}_T,\label{VOn1}
\end{align}
\begin{align}
&V_{OT}(b,\ p,\ n) = p\bigg[(1-\tau)R_2 \nonumber\\
&+ \beta\sum^{M-1}_{m=0}q_m V(\min(b+m-\mathcal{E}_T,B_{max}),\lambda_1,\ n-1) \bigg]\nonumber\\
&+(1-p)\bigg[(1-\tau)R_1\nonumber\\
&+ \beta\sum^{M-1}_{m=0}q_m V(\min(b+m-\mathcal{E}_T,B_{max}),\lambda_0,\ n-1) \bigg],\ \text{for}\ b\geq \mathcal{E}_T\label{VOT},
\end{align}
\begin{align}
&V_{OD}(b,\ p,n) \nonumber\\
&= p \beta\sum^{M-1}_{m=0}q_m V(\min(b+m-\tau\mathcal{E}_T,B_{max}),\lambda_1,\ n-1)\nonumber\\
&+(1-p) \beta\sum^{M-1}_{m=0}q_m V(\min(b+m-\tau\mathcal{E}_T,B_{max}),\lambda_0,\ n-1),\nonumber\\
&\hspace{5cm}\mbox{ for }\tau\mathcal{E}_T\leq b<\mathcal{E}_T,\label{VOn2}
\end{align}
\begin{align}
&V_{H}(b,\ p,\ n)= p\bigg[R_2\nonumber\\
& + \beta\sum^{M-1}_{m=0}q_m V(\min(b+m-\mathcal{E}_T,B_{max}),\lambda_1,\ n-1) \bigg]\nonumber\\
&+(1-p)\left[ \beta\sum^{M-1}_{m=0}q_m V(\min(b+m-\mathcal{E}_T,B_{max}),\lambda_0,\ n-1) \right],\nonumber\\
&\hspace{6.5cm} \text{for}\ b\geq \mathcal{E}_T\label{VAn}.
\end{align}

Note that when $b<\mathcal{E}_T$, we have $V(b,\ p,\ 1)=0$, and when $b\geq \mathcal{E}_T$ we have $V(b,\ p,\ 1)=\max\left\{R_1,\ p R_2,\ (1-\tau)pR_2, (1-\tau)[pR_2+(1-p)R_1]\right\}$ which is a maximum of four convex functions. We see that $V(b,\ p,\ 1)$ is a convex function of $p$. 

Now, let us assume that $V(b,\ p,\ n-1)$ is convex in $p$ for any $b\geq 0$, then for $a\in [0,\ 1]$ we can investigate the convexity of the value function for each action separately as follows.

For deferring the transmission, i.e., $A=D$, we can write:
\begin{align}
&V_{D}\left(b,\ a p_1+(1-a)p_2, n\right)\nonumber\\
&= \beta\sum^{M-1}_{m=0}q_m V(\min(b+m,B_{max}),J(a p_1+(1-a)p_2),n-1)\nonumber\\
&=\beta\sum^{M-1}_{m=0}q_m V(\min(b+m,B_{max}), a J(p_1)+(1-a)J(p_2),n-1)\nonumber\\
&\leq a \beta\sum^{M-1}_{m=0}q_m V(\min(b+m,B_{max}),J(p_1),n-1)\nonumber\\
&\ \ + (1-a) \beta\sum^{M-1}_{m=0}q_m V(\min(b+m,B_{max}),J(p_2),n-1)\nonumber\\
&=aV_D(b,\ p_1,\ n) + (1-a)V_D(b,\ p_2,\ n)
\end{align}
Hence, $V_{D}(b,\ p,\ n)$ is convex in $p$. Similarly, consider action $L$:

\begin{align}
&V_{L}\left(b,\ a p_1+(1-a)p_2, n\right)= R_1\nonumber\\
&+\beta\sum^{M-1}_{m=0}q_m V\big(\min(b+m-\mathcal{E}_T,B_{max}),J(a p_1+(1-a)p_2),n-1\big)\nonumber\\
&=R_1+ \beta\sum^{M-1}_{m=0}q_m V\big(\min(b+m-\mathcal{E}_T,B_{max})\nonumber\\
&\hspace{3cm},a J(p_1)+(1-a)J(p_2),n-1\big)\nonumber\\
&\leq a R_1 + a \beta\sum^{M-1}_{m=0}q_m V\big(\min(b+m-\mathcal{E}_T,B_{max}),J(p_1),n-1\big)\nonumber\\
&+ (1-a) R_1 \nonumber\\
&+ (1-a) \beta\sum^{M-1}_{m=0}q_m V\big(\min(b+m-\mathcal{E}_T,B_{max}),J(p_2),n-1\big)\nonumber\\
& = aV_L(b,\ p_1,\ n) + (1-a)V_L(b,\ p_2,\ n).
\end{align}

Thus, $V_{L}(b,\ p,\ n)$ is also convex in $p$. Note that  $V_{OD}(b,\ p,\ n)$, $V_{OT}(b,\ p,\ n)$, and $V_H(b,\ p,\ n)$ are linear functions of $p$, thus they are also convex in $p$. Since the value function $V(b,\ p,\ n)$ is the maximum of five (or, in some cases two) convex functions when $b\geq \mathcal{E}_T$ ($\tau\mathcal{E}_T\leq b<\mathcal{E}_T$), it is also convex. By induction we can claim the convexity of $V(b,\ p,\ n)$ for all $n$. Since $V(b,\ p,\ n)\rightarrow V(b,\ p)$ as $n\rightarrow\infty$, $V(b,\ p)$ is also convex.
\end{proof}
\section{Proof of Lemma \ref{thm:nondecreasing-battery}}
\label{sec:proofofthm2}
\begin{proof}
We will again use induction to prove the claim for $V(b,\ p,\ n)$ defined as in Appendix \ref{sec:convex_general} as the optimal value function when the decision horizon spans $n$ stages. We have $V(b,p,1)=0$ if $b<\mathcal{E}_T$ and we have $V(b,p,1)=\max\left\{R_1,\ p R_2,\ (1-\tau)pR_2, (1-\tau)[pR_2+(1-p)R_1]\right\}$ if $b\geq \mathcal{E}_T$. Hence, $V(b,p,1)$ is trivially non-decreasing in $b$. Suppose that $V(b,p,n-1)$ is non-decreasing in $b$.
Each of the value functions given in (\ref{VDn}), (\ref{VLn}), (\ref{VOn1}), (\ref{VOT}), (\ref{VOn2}) and (\ref{VAn}) is the summation of positive weighted non-decreasing functions. Therefore, they are all non-decreasing in $b$. Since the optimal value function is the maximum of these non-decreasing functions, it is also non-decreasing in $b$ for any $n$. Similarly to Appendix \ref{sec:convex_general}, by letting $n\rightarrow\infty$, we conclude that $V(b,\ p)$ is non-decreasing in $b$.
\end{proof}
\section{Proof of Lemma \ref{thm:nondecreasing-belief}}
\label{sec:proofofthm3}
\begin{proof}
We employ induction on $V(b,\ p,\ n)$ once again. For $n=1$, $V(b,\ p,\ 1)$ is $0$ if $b<\mathcal{E}_T$, and $\max\left\{R_1,\ p R_2,\ (1-\tau)pR_2, (1-\tau)[pR_2+(1-p)R_1]\right\}$ if $b\geq \mathcal{E}_T$. Therefore, $V(b,\ p,\ 1)$ is non-decreasing in $p$ for any $b$.

Assume that $V(b,\ p,\ n-1)$ is non-decreasing in $p$. Since $J(p)$ is non-decreasing, it is easy to see that $V_D(b,\ p,\ n)$ in (\ref{VDn})  and $V_L(b,\ p,\ n)$ in (\ref{VLn}) are also non-decreasing.

Since $V_A(b,\ p,\ n)$s for $A\in\left\{OD,OT,H\right\}$ are linear in $p$, we have  $V_A(b,\ ap_1+(1-a)p_0,n)=aV_A(b,\ p_1,\ n)+(1-a)V_A(b,\ p_0,\ b)$. Using this result, we have 
\begin{subequations}
\begin{align}
&V_{A}(b,\ p_1,\ n)-V_{A}(b,\ p_0,\ n)\nonumber\\
&= V_{A}(b,\ p_1-p_0+p_0,\ n)-V_{A}(b,\ p_0,\ n)\\
&=V_{A}(b,\ p_1-p_0,\ n)\geq 0,\ \ A\in\left\{OD,OT,H\right\}   \label{VAconv}
\end{align}
\end{subequations}
Note that ($\ref{VAconv}$) follows from the fact that $V_{A}(b,\ p_1-p_0+p_0,\ n)=V_{A}(b,\ p_1-p_0,\ n)+V_{A}(b,\ p_0,\ n)$. 
Since the value function, $V(b,\ p,\ n)$, is the maximum of non-decreasing functions, it is also non-decreasing. Hence, by letting $n\rightarrow\infty$, we prove that $V(b,p)$ is non-decreasing in $p$.
\end{proof}

\section{Optimality of always transmitting in a GOOD state}
\label{reviewerblahblah}

After the sensing outcome is revealed to be in a GOOD state, the transmitter may defer, or transmit at low rate, instead of transmitting at high rate. It is easy to see that, it is suboptimal to transmit at low rate when the channel is in a GOOD state. Any low rate transmission can be replaced by a high rate transmission at no additional cost, resulting in a higher value function. To show that it is also suboptimal to defer when the channel is in a GOOD state, we need to define two new actions in addition to actions $OD$ and $OT$. We define the action $ODD$, which defers transmission whatever the channel state is, and the action $OTD$, which defers transmission after sensing a GOOD channel state, but it transmits at a low rate in a BAD state. The action-value function for actions $ODD$ and $OTD$ evolve as follows:

\begin{align}
&V_{ODD}(b,\ p)= p\left[\beta\sum^{M-1}_{m=0}q_m V(\min(b+m-\tau\mathcal{E}_T,B_{max}),\lambda_1) \right]\nonumber\\
&+(1-p)\left[ \beta\sum^{M-1}_{m=0}q_m V(\min(b+m-\tau\mathcal{E}_T,B_{max}),\lambda_0) \right],\\
&V_{OTD}(b,\ p)= p\left[\beta\sum^{M-1}_{m=0}q_m V(\min(b+m-\tau\mathcal{E}_T,B_{max}),\lambda_1) \right]\nonumber\\
&+(1-p)\left[(1-\tau)R_1+ \beta\sum^{M-1}_{m=0}q_m V(\min(b+m-\mathcal{E}_T,B_{max}),\lambda_0) \right].
\end{align}


We will show that it is optimal to transmit after sensing a GOOD channel state by proving that $V_{OD}(b,\ p)>V_{ODD}(b,\ p)$ and $V_{OT}(b,\ p)>V_{OTD}(b,\ p)$, $\forall\  b,\ p$. First, we need the following lemma.

\begin{lemma}
\label{thm:value_bound}
	 For $b\geq 1$ and any $0\leq p\leq 1$, $V(b+(1-\tau)\mathcal{E}_T,\ p)-V(b,\ p)< (1-\tau)R_2$.
\end{lemma}
\begin{proof}
We will use induction to prove the lemma, and define $V(b,\ p,\ n)$ as in Appendix \ref{sec:convex_general}. For $n=1$, we have $V(b+(1-\tau)\mathcal{E}_T,\ p,\ 1)-V(b,\ p,\ 1)=0$. Assume that the lemma holds for $n-1$. We need to show that the lemma also holds for $n$. We will prove that $V_{A_1}(b+(1-\tau)\mathcal{E}_T,\ p,\ n)-V_{A_2}(b,\ p,\ n)\leq (1-\tau)R_2$ for $A_1,A_2\in \mathcal{A}_G$, where $\mathcal{A}_G=\left\{D,\ L,\ OD,\ ODD,\ OT,\ OTD,\ H\right\}$. 

Let us assume that at both states $(b+(1-\tau)\mathcal{E}_T,\ p,\ n)$ and $(b,\ p,\ n)$ it is optimal to choose action $D$. We have
\begin{align}
&V(b+(1-\tau)\mathcal{E}_T,\ p,\ n)-V(b,\ p,\ n)\nonumber\\
&=V_{D}(b+(1-\tau)\mathcal{E}_T,\ p,\ n)-V_{D}(b,\ p,\ n)\nonumber\\
&=\beta\sum^{M-1}_{m=0}q_m V(\min(b+(1-\tau)\mathcal{E}_T+m,B_{max}),J(p),\ n-1)\nonumber\\
&-\beta\sum^{M-1}_{m=0}q_m V(\min(b+m,B_{max}),J(p),\ n-1)\nonumber\\
&<\beta \sum^{M-1}_{m=0}q_m (1-\tau)R_2 = \beta (1-\tau)R_2<(1-\tau)R_2.\label{D_bound}
\end{align}
Let us assume that at states $(b+(1-\tau)\mathcal{E}_T,\ p,\ n)$ and $(b,\ p,\ n)$ it is optimal to choose the action $L$. We have
\begin{align}
&V(b+(1-\tau)\mathcal{E}_T,\ p,\ n)-V(b,\ p,\ n)\nonumber\\
&=V_{L}(b+(1-\tau)\mathcal{E}_T,\ p,\ n)-V_{L}(b,\ p,\ n)\nonumber\\
&=\beta\sum^{M-1}_{m=0}q_m V(\min(b+(1-\tau)\mathcal{E}_T+m-\mathcal{E}_T,B_{max}),J(p),\ n-1)\nonumber\\
&-\beta\sum^{M-1}_{m=0}q_m V(\min(b+m-\mathcal{E}_T,B_{max}),J(p),\ n-1)\nonumber\\
&<\beta \sum^{M-1}_{m=0}q_m (1-\tau)R_2 = \beta (1-\tau)R_2<(1-\tau)R_2.\label{L_bound}
\end{align}
Similarly, it follows that  $V_{A}(b+(1-\tau)\mathcal{E}_T,\ p,\ n)-V_{A}(b,\ p,\ n)\leq (1-\tau)R_2$ for $A\in \left\{ OD,\ ODD,\ OT,\ OTD,\ H\right\}$. 

Next, we consider cases when different actions are optimal for the two state. First we assume that it is optimal to choose action $D$ at state $(b+(1-\tau)\mathcal{E}_T,\ p,\ n)$, and  action $L$ at state $(b,\ p,\ n)$. We can write
\begin{align}
&V(b+(1-\tau)\mathcal{E}_T,\ p,\ n)-V(b,\ p,\ n)\nonumber\\
&=V_{D}(b+(1-\tau)\mathcal{E}_T,\ p,\ n)-V_{L}(b,\ p,\ n)\nonumber\\
&= V_{D}(b+(1-\tau)\mathcal{E}_T,\ p,\ n)-V_{D}(b,\ p,\ n)\nonumber\\
&+ V_{D}(b,\ p,\ n) - V_{L}(b,\ p,\ n)\nonumber\\
&<(1-\tau)R_2 + 0 = (1-\tau)R_2\label{dif_action},
\end{align}
where (\ref{dif_action}) follows since $L$ is the optimal action at state $(b,\ p,\ n)$; and hence, $V_{D}(b,\ p,\ n) - V_{L}(b,\ p,\ n)\leq 0$. Also,  $V_{D}(b+(1-\tau)\mathcal{E}_T,\ p,\ n)-V_{D}(b,\ p,\ n)<(1-\tau)R_2$ as we have shown in (\ref{D_bound}). 

Similar to the derivations of (\ref{dif_action}), we can easily prove that $V_{A_1}(b+(1-\tau)\mathcal{E}_T,\ p,\ n)-V_{A_2}(b,\ p,\ n)\leq (1-\tau)R_2$ for $A_1\in \mathcal{A}_G$ and $A_2\in \left\{\mathcal{A}_G\backslash A_1\right\}$.

Combining all the above results, we can finally state that $V(b+(1-\tau)\mathcal{E}_T,\ p,\ n)-V(b,\ p,\ n)<(1-\tau)R_2$. Since $V(b,\ p,\ n)\rightarrow V(b,\ p)$ as $n\rightarrow\infty$, we have $V(b+(1-\tau)\mathcal{E}_T,\ p)-V(b,\ p)<(1-\tau)R_2$.
\end{proof}
In the following, we will show that $V_{OD}(b,\ p)>V_{ODD}(b,\ p)$. We have
\begin{subequations}
\begin{align}
&V_{OD}(b,\ p)-V_{ODD}(b,\ p) = p(1-\tau)R_2\nonumber\\
&+ p\beta \sum^{M-1}_{m=0}q_m \big[V(\min(b+m-\mathcal{E}_T,B_{max}),\lambda_1)\nonumber\\
&\hspace{1.9cm}-V(\min(b+m-\tau\mathcal{E}_T,B_{max}),\lambda_1) \big]\label{son1}\\
&>p(1-\tau)R_2-p\beta \sum^{M-1}_{m=0}q_m(1-\tau)R_2 =p(1-\beta)(1-\tau)R_2>0,\label{son2}
\end{align}
\end{subequations}
where we use the result established in Lemma \ref{thm:value_bound} to simplify (\ref{son1}) into (\ref{son2}). With the same outline in the above, it directly follows that $V_{OT}(b,\ p)>V_{OTD}(b,\ p)$. The intuition behind the above result is the fact that by saving $(1-\tau)\mathcal{E}_T$ units of energy in the GOOD state, one cannot get a better reward than $(1-\tau)R_2$ in the future.  Hence, there is no reason to save the energy when we are sure that the channel is in a GOOD state.

\end{appendices}

\bibliographystyle{unsrt} 

\bibliography{Bibliography}

\begin{IEEEbiography}[{\includegraphics[width=1in,height=1.25in,clip,keepaspectratio]{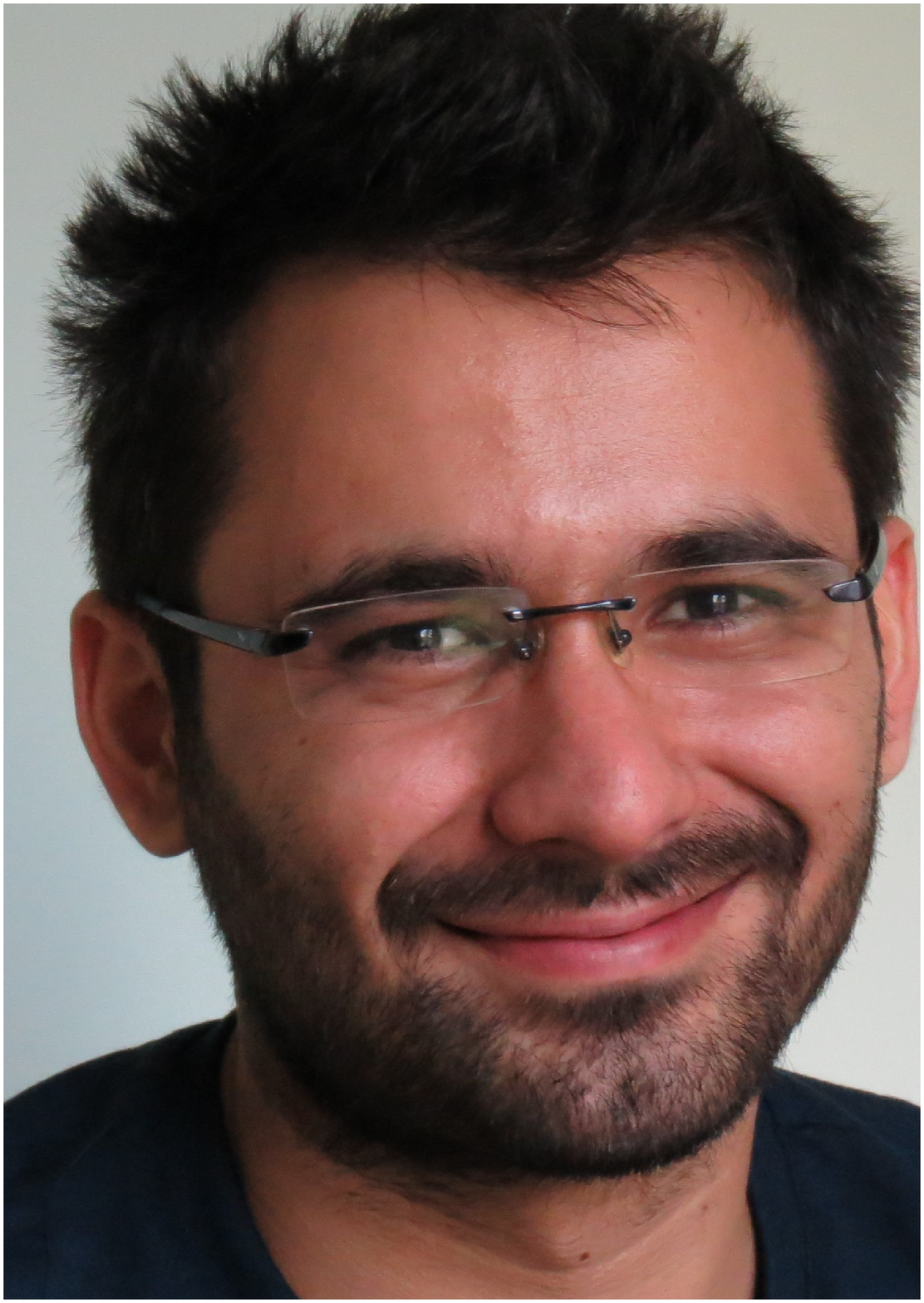}}]{Mehdi Salehi Heydar Abad}
received the B.S. degree in electrical engineering from the IUST, Tehran, Iran, in 2012, and the M.S. degree in electrical and electronics engineering from the Sabanci University, Istanbul, Turkey in 2015. Currently he is a PhD candidate at Sabanci University. He was also a Visiting Researcher at The Ohio State University, Columbus, OH, USA. His research interests are in the field of mathematical modeling of communication systems, stochastic optimization, and green communication networks. He is the recipient of the Best Paper Award at the 2016 IEEE Wireless Communications and Networking Conference (WCNC).
\end{IEEEbiography}

\begin{IEEEbiography}[{\includegraphics[width=1in,height=1.25in,clip,keepaspectratio]{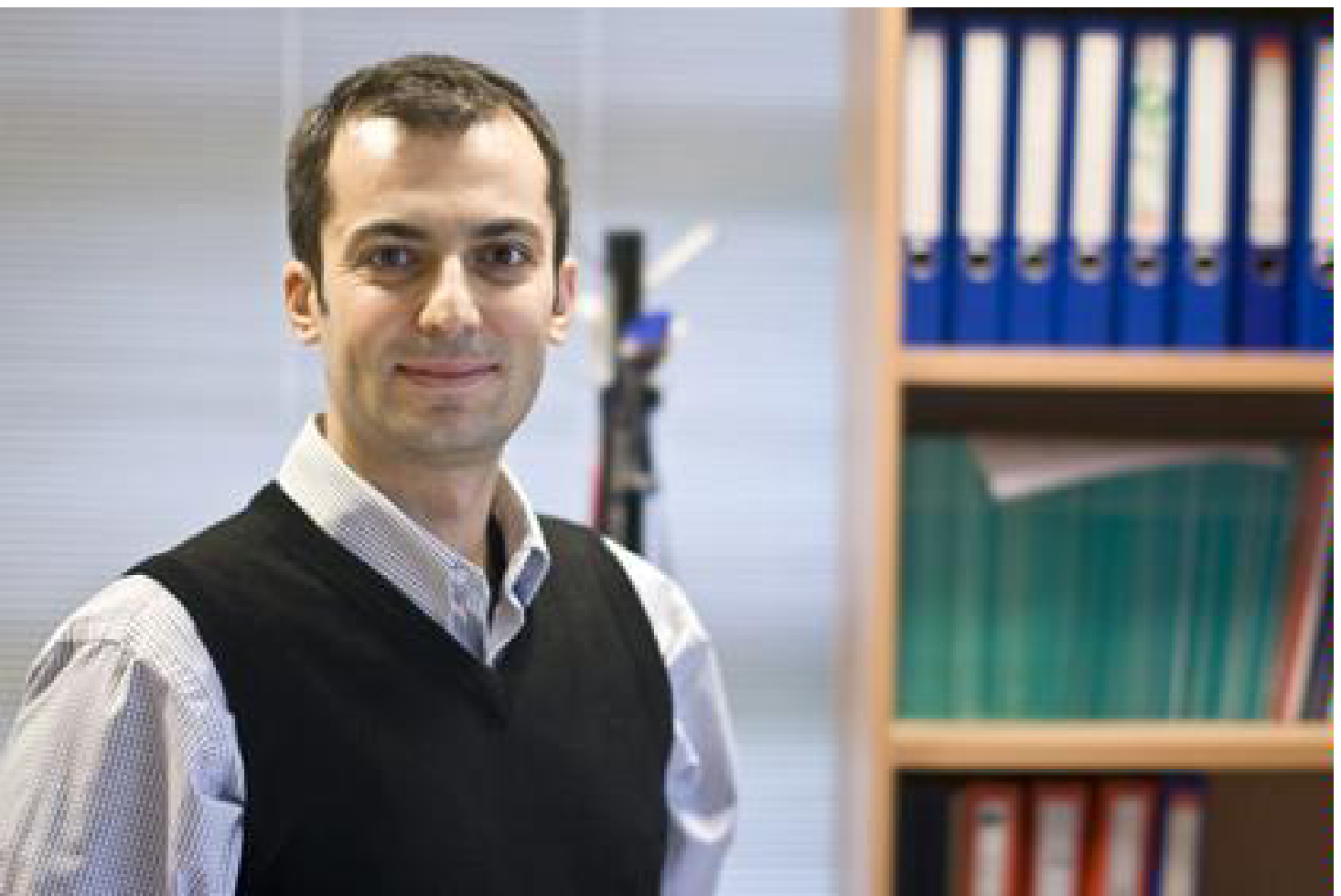}}]{Ozgur Ercetin}
received the B.S. degree in electrical and electronics engineering from the Middle East Technical University, Ankara, Turkey, in 1995, and the M.S. and Ph.D. degrees in electrical engineering from the University of Maryland, College Park, MD, USA, in 1998 and 2002, respectively. Since 2002, he has been with the Faculty of Engineering and Natural Sciences, Sabanci University, Istanbul, Turkey. He was also a Visiting Researcher with HRL Labs, Malibu, CA, USA; Docomo USA Labs, Palo Alto, CA, USA; The Ohio State University, Columbus, OH, USA; Carleton University, Ottawa, CA and Université du Québec à Montréal, Montreal CA. His research interests are in the field of computer and communication networks with emphasis on fundamental mathematical models, architectures and protocols of wireless systems, and stochastic optimization.
\end{IEEEbiography}

\begin{IEEEbiography}[{\includegraphics[width=1.5in,height=1.25in,clip,keepaspectratio]{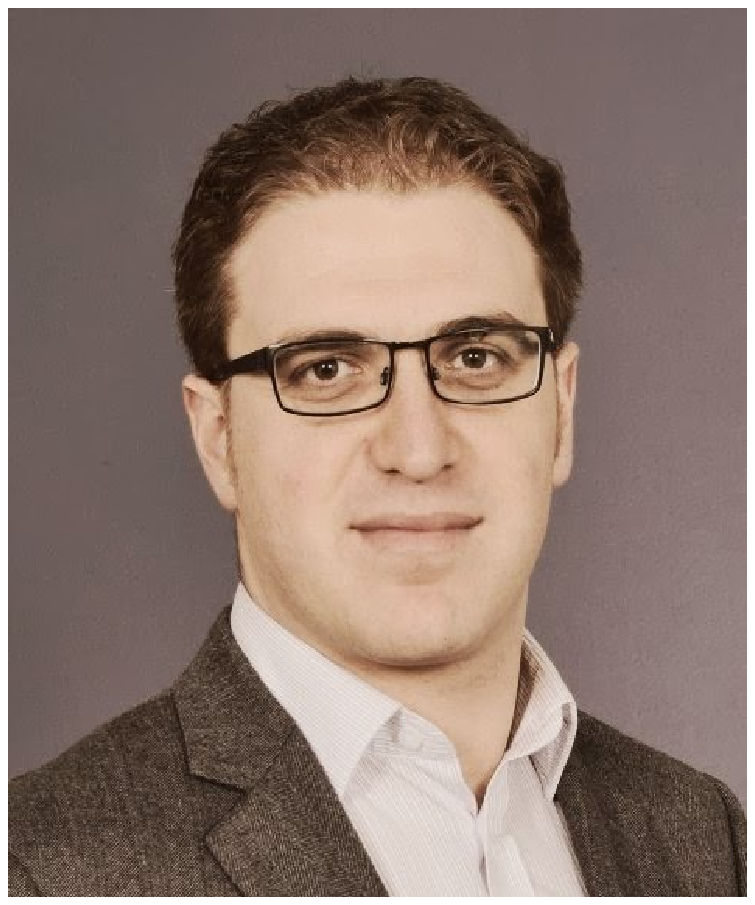}}]{Deniz G\"{u}nd\"{u}z}
[S’03-M’08-SM’13] received the B.S. degree in electrical and electronics engineering from METU, Turkey in 2002, and the M.S. and Ph.D. degrees in electrical engineering from NYU Polytechnic School of Engineering (formerly Polytechnic University) in 2004 and 2007, respectively. After his PhD, he served as a postdoctoral research associate at Princeton University, and as a consulting assistant professor at Stanford University. He also served as a research associate at CTTC in Spain, until he joined the Electrical and Electronic Engineering Department of Imperial College London, UK, where he is currently a Reader in information theory and communications.
His research interests lie in the areas of communications and information theory, machine learning, and security and privacy in cyber-physical systems. Dr. G\"{u}nd\"{u}z is an Editor of the IEEE TRANSACTIONS ON COMMUNICATIONS, and the IEEE TRANSACTIONS ON GREEN COMMUNICATIONS AND NETWORKING. He is the recipient of the IEEE Communications Society - Communication Theory Technical Committee (CTTC) Early Achievement Award in 2017, Starting Grant of the European Research Council (ERC) in 2016, and the IEEE Communications Society Best Young Researcher Award for the Europe, Middle East, and Africa Region in 2014. He is currently serving as the General Co-chair of the 2018 Workshop on Smart Antennas.
\end{IEEEbiography}
\end{document}